\pdfoutput=1
\documentclass[11pt]{article}

\usepackage[a4paper,left=20mm,right=20mm,top=30mm,bottom=20mm]{geometry}
\usepackage{lmodern}
\usepackage{microtype}
\usepackage[square,authoryear]{natbib}
\usepackage{float}
\usepackage{titletoc}

\IfFileExists{arXiv/headers.tex}{\def\arxivroot{arXiv/}}{\def\arxivroot{}}
\makeatletter
\edef\input@path{{\arxivroot}}
\makeatother

\usepackage{amsmath,amsfonts,bm}

\def\eqref#1{equation~\ref{#1}}

\def\1{\bm{1}}

\DeclareMathAlphabet{\mathsfit}{\encodingdefault}{\sfdefault}{m}{sl}
\SetMathAlphabet{\mathsfit}{bold}{\encodingdefault}{\sfdefault}{bx}{n}

\usepackage[T1]{fontenc}
\usepackage{amsfonts, amssymb, amsmath, amsthm, graphicx, color, tikz, booktabs, multirow, bm, cases, mathtools,  mathrsfs, pgfplots, subcaption, csvsimple,algorithm,algorithmic}
\usetikzlibrary{intersections, pgfplots.fillbetween, shapes, arrows, positioning, decorations.markings}
\usetikzlibrary{matrix, fit, shapes}
\definecolor {processblue}{cmyk}{0.96,0,0,0}
\tikzstyle{int}=[draw, fill=blue!20, minimum size=2em]
\tikzstyle{init} = [pin edge={to-,thin,black}]
\usepackage{pgfplotstable}
\usepackage{pgfplots}
\pgfplotsset{compat=1.14}

\newtheorem{rem}{Remark}[section]

\newtheorem{thm}{Theorem}[section]

\newcommand{\ex}[2]{{\ifx&#1& \mathbb{E} \else {\mathbb{E}_{#1}} \fi \left[#2\right]}}
\newcommand{\indp}{\perp\!\!\!\perp} 

\newcommand{\br}[1]{\left[#1\right]}
\newcommand{\abs}[1]{\left|#1\right|}

\newcommand{\kl}[2]{\mathrm{D_{KL}}\left(#1 \|#2 \right)}

\usepackage[toc,page,header]{appendix}
\usepackage{enumitem}
\usepackage{thm-restate}

\newcommand{\calA}{\mathcal A}

\newcommand{\calW}{\mathcal W}

\newcommand{\calZ}{\mathcal Z}
\newcommand{\CMI}{\mathrm{CMI}}
\newcommand{\eCMI}{\mathrm{eCMI}}

\newcommand{\TV}{\mathrm{TV}}
\newcommand{\Bern}{\mathrm{Bern}}
\newcommand{\Unif}{\mathrm{Unif}}

\newcommand{\ones}{\mathbf 1}

\newcommand{\defeq}{\coloneqq}

\newcommand{\eqnref}[1]{Eq.~(\ref{#1})}

\usepackage{xcolor}
\usepackage[normalem]{ulem}

\usepackage{hyperref}
\usepackage{url}
\usepackage[capitalize,noabbrev]{cleveref}
\graphicspath{{arXiv/}{./}}
\hypersetup{hidelinks,
  pdftitle={Learning Rules, Memorizing Facts: An Information-Theoretic View of Hallucination},
  pdfauthor={Gengze Xu, Bob Guo, Yifei Hu, Xiangyu Wu, Yong Liu, Ziqiao Wang}}
\title{Learning Rules, Memorizing Facts:\\An Information-Theoretic View of Hallucination}

\author{\textbf{Gengze Xu}\textsuperscript{1}\quad
\textbf{Anxin Guo}\textsuperscript{2}\quad
\textbf{Yifei Hu}\textsuperscript{3}\quad
\textbf{Xiangyu Wu}\textsuperscript{3}\\[3pt]
\textbf{Yong Liu}\textsuperscript{1}\quad
\textbf{Ziqiao Wang}\textsuperscript{4}\\[8pt]
\small\textsuperscript{1}Renmin University of China\quad
\small\textsuperscript{2}Northwestern University\\
\small\textsuperscript{3}Kuaishou Technology\quad
\small\textsuperscript{4}Tongji University\\[6pt]
\texttt{\footnotesize xugengze@ruc.edu.cn\quad bobguo2023@u.northwestern.edu}\\
\texttt{\footnotesize huyifei@kuaishou.com\quad wuxiangyu06@kuaishou.com}\\
\texttt{\footnotesize liuyonggsai@ruc.edu.cn\quad ziqiaowang@tongji.edu.cn}}
\date{}

\begin{document}

\maketitle

\begin{abstract}
Understanding the relationship between generalization and hallucination requires clarifying the role of memorization. We study this relationship from an information-theoretic perspective through a Selector-Density Conditional Mutual Information (CMI) framework. In the sparse random-fact limit, a uniformly bounded information budget per fact forces a nonvanishing false-positive rate at recall bounded away from zero, even when all facts are observed during training. In a rules-and-facts setting, we decompose model–data mutual information into rule learning, atomic-fact memory, and excess memorization, revealing competition between rule learning and factual memory even without excess memorization. A joint budget bound quantifies the resulting constraint on atomic hallucination under perfect atomic recall. Moreover, an additive Gaussian model shows that maintaining high recall for atomic and composite facts incurs false positives even under optimal data mixing. Finally, controlled relational experiments hint how training exposure, rule complexity, and model capacity shape these trade-offs.
\end{abstract}

\section{Introduction}

Large language models (LLMs) exhibit broad capabilities across language, coding, and reasoning tasks,  
along with impressive generalization abilities such as zero-shot~\citep{sanh2022multitask}, few-shot~\citep{brown2020language}, and in-context learning~\citep{garg2022what}.
Yet this ability to generalize does not guarantee factuality: a high-confidence claim about an unseen case may be either a valid inference or a hallucination, namely a plausible but false statement. 
Such hallucinations can propagate false premises into downstream reasoning and decisions. 
Recent theoretical work connects hallucination to Is-It-Valid (IIV) classification~\citep{kalai2025language,kalai2026evaluating} or membership testing~\citep{guo2026hallucination}.
These perspectives raise a fundamental question about the information retained during training: \emph{what must a model learn or memorize to distinguish facts from plausible non-facts?}

Memorization plays a central, but nuanced, role in this question. Generalization theory often controls the gap between training and population risks by limiting dependence on the particular training sample. At the same time, memorization can be necessary for accurate prediction, especially for rarely observed subpopulations and facts with little reusable structure~\citep{feldman2020does,brown2021memorization}. Learning a general rule cannot replace retaining the arbitrary facts to which that rule applies. This motivates distinguishing atomic-fact memory, which preserves factual information not determined by reusable rules, from excess memorization, which records which examples happened to be observed after the underlying factual world is fixed. \emph{How do these distinct forms of retained information jointly constrain generalization and factuality under a limited information budget?}

We address this question through a Selector-Density Conditional Mutual Information (CMI) framework. CMI measures how much a learned model reveals about the training-set selection when conditioned on a supersample of candidate examples~\citep{steinke2020reasoning}. Using a uniform selection of $n$ examples from $N$ candidates, we connect selector information to the distinguishability of evaluation outputs on selected and unselected examples: the information an individual output reveals about selection equals a weighted Jensen-Shannon divergence between its two conditional laws. The selector and evaluator determine the interpretation. When selection denotes training membership and the evaluator is a loss, small selector information controls the expected generalization gap. When selection denotes truth and the evaluator is a factuality score, accurate discrimination requires sufficient selector information to separate facts from non-facts. 

We first make the latter role explicit in a sparse random-fact setting, following the membership-testing formulation of~\cite{guo2026hallucination}. A uniformly selected subset of candidate claims constitutes the entire factual world and is fully observed during training, while all remaining claims are false. Since truth has no inferable structure, factual discrimination depends on retaining information about this subset. We derive lower bounds on the false-positive hallucination rate at prescribed fact recall. In the sparse-fact limit, a uniformly bounded information budget per fact prevents false positives from vanishing whenever recall remains bounded away from zero. Thus, post-hoc thresholding alone cannot overcome insufficient factual memory. The contrast with generalization concerns the information being measured, that is, indistinguishability of training and held-out losses controls a generalization gap, whereas indistinguishability of factual and non-factual scores obstructs factual discrimination. In this sense, hallucination can be viewed as an ``unwanted generalization'' of factual confidence from observed facts to unseen non-facts.

To incorporate knowledge that supports inference beyond observed examples, we introduce a rules-and-facts setting in which reusable rules and atomic facts jointly determine the factual world. Here, truth and training membership are distinct: an unobserved claim may still be true. We decompose model-data mutual information into rule learning, atomic-fact memory conditional on the rules, and excess memorization of the training selection. This separates the information supporting factual knowledge from residual dependence 
on the observed sample. Under a shared information budget, information committed to rule recovery and training selection limits what remains for atomic-fact memory. Consequently, even accurate rule learning and a small generalization gap do not by themselves ensure factual discrimination; namely, 
rule-independent facts still require memory. Our bounds quantify the constraint on atomic false positives under perfect atomic recall.

Our contributions are summarized as follows:
\begin{itemize}[leftmargin=1.25em]
    \item We connect the CMI generalization framework with information requirements for factual discrimination. In the sparse random-fact setting, 
    we show that a bounded information budget per fact forces a nonvanishing hallucination rate at recall bounded away from zero, even when every true fact is observed during training (Theorem~\ref{thm:fpr-random-fact}). Post-hoc thresholding therefore cannot compensate for insufficient factual memory.
    To the best of our knowledge, this is the first application of the CMI framework to hallucination analysis.

    \item In a rules-and-facts setting, we separate model-data mutual information into rule learning, atomic-fact memory and excess memorization (Lemma~\ref{3MI}). In the sparse atomic-fact regime, our joint budget bound quantifies how information committed to rule recovery and training selection constrains atomic hallucination suppression under perfect recall (Theorem~\ref{thm:joint-information-budget}). Thus, accurate rule learning and a small generalization gap do not by themselves ensure factuality.

    \item In a shared-location Gaussian model, we show how the training mixture redistributes recall between atomic and composite facts (Theorem~\ref{thm:gaussian-recall}), while their separation and dispersion impose a hallucination cost even under optimal mixing  (Corollary~\ref{cor:optimal-mixture-recall}). Controlled relational experiments show that adding composite facts during training need not improve atomic factual discrimination, and that greater rule complexity can increase its capacity requirements. Larger models ease these trade-offs, providing qualitative evidence that the factual benefits of additional structured data depend on the capacity available to retain atomic facts.
\end{itemize}

\paragraph{Other Related Literature.} Building on CMI~\citep{steinke2020reasoning}, functional CMI~\citep{harutyunyan2021informationtheoretic}, evaluated CMI~\citep{hellstrom2022a}, and loss-difference CMI~\citep{wang2023tighter} use successive post-processing to obtain tighter generalization bounds.
Leave-one-out~\citep{haghifam2022understanding,rammal2022on} and leave-$m$-out CMI~\citep{lu2026unified} further vary the sampling mechanism. These works establish CMI as a well-developed framework for generalization. However, its implications for hallucination remain unexplored. Existing analyses of hallucination instead draw on Good--Turing estimation~\citep{kalai2024calibrated}, Kolmogorov complexity~\citep{pan2025understanding}, and statistical learning theory~\citep{kalavasis2025limits}. 
Our distinction between factual knowledge and residual sample dependence also relates to excess memorization~\citep{feldman2025trade} and unintended memorization~\citep{morris2026how}.
Moreover, knowledge acquisition under finite capacity has been studied through data pruning~\citep{ye2026cram} and data mixing~\citep{gu2026data}. Appendix~\ref{sec:related-work} provides further discussion on related work.

\section{Preliminaries}

\textbf{Notations.} Throughout this paper, unless otherwise stated, capital letters (e.g., $X$) denote random variables, while the corresponding lowercase letters (e.g., $x$) denote their realizations. Expectations are expressed as $\mathbb{E}_X[\cdot]$. 
All logarithms are natural and all information quantities are in nats.  Let the binary entropy be $h(a)\coloneqq -a\log a-(1-a)\log(1-a)$. For probability measures $P,Q$ on a common measurable space, $\kl{P}{Q}$ denotes  the KL divergence.

\textbf{Supersample Setting.} Let $\mu$ be an unknown distribution over the data instance space $\mathcal{Z}$. Let $\widetilde Z=(\widetilde Z_1,\ldots,\widetilde Z_N)\overset{i.i.d.}{\sim}\mu^{\otimes N}$ be a supersample of candidate instances. Each $\widetilde Z_j$ may represent a data--label pair in supervised learning or a plausible claim in a generative-model setting. For $n\in[1,N)$, we draw a selector $K\sim\operatorname{Unif}\{A\subseteq[N]:|A|=n\}$ independently of $\widetilde Z$ and form the training sample $S:=\widetilde Z_K$. The remaining $m:=N-n$ instances constitute the holdout sample, and we define the selector density as $\rho:=\frac{n}{N}$. Note that this construction differs from the original CMI framework~\citep{steinke2020reasoning}, which arranges $2n$ instances into $n$ pairs and independently selects one instance from each pair with equal probability. The distinction lies in the selector's support, that is, the original selector is uniform over the $2^n$ subsets containing exactly one instance per pair, whereas ours is uniform over all $\binom{N}{n}$ subsets of size $n$. Thus, even when $N=2n$, the two constructions generally induce different selector laws. Our construction is closer to the leave-one-out or leave-$m$-out setting~\citep{haghifam2022understanding,rammal2022on,lu2026unified}.

\textbf{CMI and eCMI.} Let $\mathcal{W}$ be the hypothesis space, and let $\mathcal{A}:\mathcal{Z}^n\to\mathcal{W}$ be a learning algorithm that maps the training sample $S$ to a model $W:=\mathcal A(S)$. Given a measurable evaluator
$q:\mathcal W\times\mathcal Z\to\mathbb{R}$, define the evaluation outputs  $T^N:=(T_1,\ldots,T_N)$ where $T_i:=q(W,\widetilde Z_i)$. 
We define the \emph{selector-density CMI} and
its \emph{evaluated counterpart} as $\CMI_{\rho}(\calA)\coloneqq  I(W;K\mid\widetilde Z)$, and $\eCMI^q_{\rho}(\calA)\coloneqq  I(T^N;K\mid\widetilde Z)$. These quantities measure how much the model and its evaluated behavior,
respectively, reveal about the training-set selection given the supersample. Since 
$T^N$ is a deterministic function of $(W,\widetilde Z)$, the data-processing inequality gives $\eCMI^q_{\rho}(\mathcal A)\leq\CMI_{\rho}(\mathcal A)$.

At the individual level, let $U_i\coloneqq\ones\{i\in K\}$ denote the membership indicator of $\widetilde Z_i$. Uniform subset selection implies $U_i\sim\Bern(\rho)$. The mutual information $I(U_i;T_i)$ quantifies how much the evaluation output $T_i$ reveals about whether $\widetilde Z_i$ was selected for training, providing an information-theoretic perspective on membership testing~\citep{steinke2020reasoning}.

For a fixed coordinate $i\in[N]$, define the selected-item and unselected-item laws $P\coloneqq\operatorname{Law}(T_i\mid U_i=1)$ and $Q\coloneqq\operatorname{Law}(T_i\mid U_i=0)$, where we omit their dependence on $i$ for notational simplicity. The mixture $M_\rho\defeq\rho P+(1-\rho)Q$ is the marginal law of $T_i$. The corresponding weighted
Jensen--Shannon (JS) divergence is $J_\rho(P,Q)\defeq\rho\kl{P}{M_\rho}+(1-\rho)\kl{Q}{M_\rho}$. 
The following lemma connects this divergence to individual
selector information.

\begin{restatable}{lem}{selectorJSidentity}\label{cmi-js}
Under the preceding notation, for any fixed $i\in[N]$, we have
\begin{align}\label{eq:cmi-js}
    I(U_i;T_i)=J_\rho(P,Q).
\end{align}
\end{restatable}

\begin{rem}
\label{rem:cmi-membership}
\eqnref{eq:cmi-js} identifies the information that $T_i$ reveals about
membership with the prior weighted JS divergence between the two conditional output laws. Equivalently, $I(U_i;T_i)=h(\rho)-H(U_i\mid T_i)$, where $h$ denotes binary entropy. Thus, selector information measures the
expected reduction in membership uncertainty after observing $T_i$.
If $P=Q$, then $I(U_i;T_i)=0$, that is, the output provides no membership information beyond the prior. At the opposite extreme, if $P$ and $Q$ are mutually
singular, membership can be recovered perfectly from $T_i$, and
$I(U_i;T_i)=h(\rho)$. The selector density therefore determines 
the maximum information that
an individual output can reveal. 
\end{rem}

At the dataset level, CMI and eCMI share the same uncertainty-reduction interpretation. 
Conditioning on the supersample fixes the candidate instances, so the remaining uncertainty concerns which instances were selected.
This membership-testing perspective underlies the analyses that follow.

\section{CMI View of Generalization and Hallucination}
\label{sec:cmi-gen-hall}

This section studies generalization and random-fact hallucination through the membership-testing interpretation of selector information.

\subsection{Generalization Bounds}
\label{subsec:cmi-generalization}

We revisit CMI generalization analysis, deriving individual and global bounds that illustrate how the membership information revealed by model losses controls the expected generalization gap.

Let the evaluator $q$ be a bounded loss function
$\ell:\calW\times\calZ\to[0,1]$, so that
$T_i=\ell(W,\widetilde Z_i)$. Let $K^c\defeq[N]\setminus K$, and define the empirical and held-out risks as $\hat L_n\defeq\frac{1}{n}\sum_{i\in K}\ell(W,\widetilde Z_i)$ and $\hat L_m\defeq\frac{1}{m}\sum_{i\in K^c}\ell(W,\widetilde Z_i)$, respectively. Let
$L_\mu(w)\defeq\ex{Z\sim\mu}{\ell(w,Z)}$ denote the population risk. Under the i.i.d. supersample construction, the held-out observations are independent of $W$ conditional on $K$, and hence $\ex{}{\hat L_m}=\ex{}{L_\mu(W)}$. The expected generalization error can therefore be written as $\mathcal E(\calA)\defeq\ex{\widetilde Z,K,W}{\hat L_m-\hat L_n}=\ex{}{L_\mu(W)-\hat L_n}$. Our first result bounds this gap through the selector information
associated with an individual evaluation output.

\begin{restatable}{thm}{randomCoordinateBound}\label{thm:random-coordinate}
Let $J\sim\Unif([N])$ be independent of
$(\widetilde Z,K,W)$, and define
$U\defeq\ones\{J\in K\}$ and $T\defeq T_J$. Then
\begin{equation}
\abs{\mathcal E(\calA)}
\le
\sqrt{\frac{I(U;T\mid\widetilde Z,J)}{2\rho(1-\rho)}}.
\label{eq:random-coordinate-bound}
\end{equation}
\end{restatable}

The next result instead uses the global evaluated CMI of the entire loss vector, $\eCMI^\ell_\rho(\calA)$.
\begin{restatable}{thm}{globalECMIBound}
\label{thm:global-ecmi}
Under the construction above, we have
\begin{align}
\label{eq:global-ecmi}
|\mathcal{E}(\mathcal A)|
\leq
\min\left\{
\frac{N}{m}\sqrt{\frac{\eCMI^\ell_\rho(\mathcal A)}{2n}},
\frac{N}{n}\sqrt{\frac{\eCMI^\ell_\rho(\mathcal A)}{2m}}
\right\}.
\end{align}
\end{restatable}

The two theorems quantify selector information at different levels. Theorem~\ref{thm:random-coordinate} uses an average of coordinate-wise CMI, whereas Theorem~\ref{thm:global-ecmi} uses the information that the full loss vector reveals about the selected subset. Both provide a membership-testing interpretation, namely when the losses reveal little about which instances were selected for training, the expected training and held-out risks cannot differ significantly.
Importantly, this controls the generalization gap rather than the population risk itself, and does not require the entire model to carry little information about the training sample. 

Our bounds use a fixed-cardinality selector with general density, so the information quantities need not coincide with those of the original CMI framework. Our objective is not to improve upon existing CMI bounds, but to establish the selector-information perspective that also enables our hallucination analysis. We defer a comparison with existing CMI bounds to Appendix (see Remark~\ref{rem:existing-cmi-bounds}).

\subsection{Hallucination in the Random-Fact Setting}
\label{subsec:cmi-hallucination}

We now discuss generative models and consider a random-fact setting.
For this subsection, let $\widetilde Z=(\widetilde Z_1,\ldots,\widetilde Z_N)$ be a fixed universe of $N$ distinct plausible claims, rather than an i.i.d. supersample. Draw $K\sim\operatorname{Unif}\{A\subseteq[N]:|A|=n\}$, and declare precisely the claims indexed by $K$ to be true. The training sample $S=\widetilde{Z}_K$ therefore contains all facts, while $\widetilde Z_{K^c}$ contains all nonfacts. Thus, in this closed-world abstraction, $K$ serves both as the training-set selector and as the truth selector, following the random-fact membership-testing model of~\citet{guo2026hallucination}. Since truth is assigned uniformly at random, the claim identities provide no additional structure from which membership can be inferred, and factual discrimination must instead rely on information retained about $K$.
Recall that $J\sim\Unif([N])$, $U\defeq\ones\{J\in K\}$ and $T\defeq q(W,\widetilde Z_J)$. Here, $U\sim\Bern(\rho)$ indicates whether the queried claim is true, and $T\in[0,1]$ is a factuality score, with larger values indicating greater confidence in the claim's truth.
Accordingly, $P\defeq\operatorname{Law}(T\mid U=1)$ and $Q\defeq\operatorname{Law}(T\mid U=0)$ are now score laws on facts and non-facts, respectively, rather than loss distributions on training and held-out samples. Since the claim universe is fixed, conditioning on $\widetilde Z$ is implicit throughout this subsection (i.e. $\widetilde{Z}=\tilde{z}$).

We focus first on the sparse regime $\rho=n/N\ll1$, in which true facts
constitute a small fraction of the candidate
universe~\citep{guo2026hallucination,kalai2025language,kalai2026evaluating}.
For intuition, a birthday query with $365$ candidate dates
and exactly one correct answer has truth density $1/365$ under uniform
candidate querying. This illustrates a small truth base rate. The limit
$\rho\to0$ idealizes an increasingly large set of plausible
alternatives relative to the number of facts.

The following lemma identifies the normalized selector information in this regime.

\begin{restatable}{lem}{sparseJSKLLimit}\label{lem:js-kl-limit}
Under the setting of Lemma~\ref{cmi-js}, 
as $\rho = \frac{n}{N} \to0$,
$\frac{I(U;T)}{\rho}
=
\frac{J_\rho(P,Q)}{\rho}
\longrightarrow
\kl{P}{Q}$.
\end{restatable}

Lemma~\ref{lem:js-kl-limit}, similar to~\citet[Lemma~3.5]{guo2026hallucination}, 
shows that directional KL is the
limiting selector information per rare true claim, with $Q$ serving as
the nonfact background. This hints at a \emph{capacity-per-unit-fact} interpretation, namely for the stationary memoryless extension of the binary channel $U\mapsto T$, assign cost zero to $U=0$ and cost one to $U=1$.
Its capacity per unit cost is then $\kl{P}{Q}$ by the zero-cost-symbol formula~\citep{yury2023information}.

Under their membership-testing assumptions,
\cite{guo2026hallucination} establish that, for prescribed error
constraints, the minimum asymptotic memory rate per fact is the
minimum of $\kl{P}{Q}$ over admissible fact and non-fact score laws,
along sequences with $n\to\infty$ and $n/N\to0$.
This result connects the score laws'  distinguishability to the memorization required for random facts.

Generative models generally assign nonzero probability to unseen outputs, including non-facts, thereby incurring prediction error; for example, softmax assigns positive probability to every candidate token. Hallucination, however, is the special case in which a non-fact receives sufficient confidence to be treated as factual.
In the random-fact setting, facts with $U=1$ form the positive class, and non-facts with $U=0$ form the negative class. We consider a threshold-based notion of factuality: a candidate claim is accepted as factual when its score reaches a prescribed threshold. For $A_\tau=\mathbf{1}\{T\geq\tau\}$, define the fact recall and false-positive rate as $\alpha_\tau:=\mathbb{P}(A_\tau=1\mid U=1)$ and $\beta_\tau:=\mathbb{P}(A_\tau=1\mid U=0)$, respectively. 
Thus, accepting a nonfact is a false positive, and we refer to
$\beta_\tau$ as the \emph{false-positive hallucination rate}.
A natural question is whether a suitable threshold can eliminate false-positive hallucinations while retaining factual claims.
We characterize their relationship in the sparse-fact regime below.

\begin{restatable}{thm}{FPRrandomfact}
\label{thm:fpr-random-fact}
In the random-fact setting, suppose that $\frac{1}{n}I(T^N;K)\le R$.
Then in the sparse limit $\rho\to0$ and $n \to \infty$,
 for every $\tau$ s.t. $\alpha_\tau>0$, we have
$\beta_\tau
 \geq
 \exp\left(-\frac{R+h(\alpha_\tau)}{\alpha_\tau}\right)$.
In particular, perfect recall $\alpha_\tau=1$ implies $\beta_\tau\geq e^{-R}$.
\end{restatable}

Theorem~\ref{thm:fpr-random-fact} identifies an information-limited
recall--hallucination trade-off, illustrated in
Figure~\ref{fpr-recall-bound} in Appendix~\ref{proofs}. Under a uniformly bounded score-level
information budget per fact, recall bounded away from zero prevents the
hallucination rate from vanishing in the sparse limit.
Changing the threshold can move the operating point, but cannot bypass
this constraint for the same score array. In addition, in the random-fact setting, facts contain no reusable structure, so factual discrimination reduces to memorizing selector $K$. Since $I(T^N;K)\leq I(W;K)\leq H(W)$, finite storage therefore provides one sufficient source of the information bottleneck. As a by-product of our analysis, we also study the relationship between calibration and hallucination in Appendix~\ref{sec:appendix-calibration}.

\begin{rem}[Revisiting Lemma~\ref{cmi-js}: two roles of selector information]
\label{rem:selector-gen-hall}
Lemma~\ref{cmi-js} gives both (generalization and hallucination) analyses a membership-testing interpretation through $I(U;T)=J_\rho(P,Q)$.
Small selector information in \emph{losses} controls the generalization gap, whereas high fact recall with few false positives requires sufficient selector information in \emph{factuality scores}. Indeed, $P=Q$ implies $\alpha_\tau=\beta_\tau$ for every threshold. 
In the random-fact setting, hallucination can be viewed as \emph{unwanted generalization of factual confidence} from observed facts to unseen nonfacts.
\end{rem}

The random-fact abstraction isolates the role of memorization.
Real knowledge, however, also contains reusable structure that supports
inference beyond observed examples.
In the next section, we introduce a rules-and-facts setting to study
how rule learning and atomic-fact memory jointly constrain factual
discrimination under a shared information budget.

\section{Rules-and-Facts Setting}
\label{sec:rules-and-facts}

We now extend the random-fact setting to a rules-and-facts model in which atomic facts and reusable rules jointly generate the factual world. This construction distinguishes rule learning from atomic-fact memorization at the level of the generative model and its information-theoretic analysis.

Formally, let $\mathscr B=(b_1,\ldots,b_{N_1})$ be a fixed pool of $N_1$ different candidate atomic claims. An atomic fact selector $K_1\sim\operatorname{Unif}\{A\subseteq[N_1]:|A|=n_1\}$ determines the true atomic facts $B\coloneqq\mathscr B_{K_1}$, with $\rho_1\coloneqq \frac{n_1}{N_1}$. Let $\Theta\sim P_\Theta$ represent the reusable rule structure, independently of $K_1$.
Given $(B,\Theta)$, a deterministic mapping $\Gamma$ generates the full factual supersample
$\widetilde Z=\Gamma(B,\Theta)=(\widetilde Z_1,\ldots,\widetilde Z_N)$, which contains all true atomic and composite facts within the modeled universe. 
Unlike the earlier i.i.d. setting,  the entries $\widetilde Z_i$ need not be independent, since they are generated from the shared latent variables $(B,\Theta)$. Independently, let $K\sim\operatorname{Unif}\{A\subseteq[N]:|A|=n\}$ select the training set $S\coloneqq\widetilde Z_K$, with sampling density $\rho\coloneqq \frac{n}{N}$. The learner then outputs $W=\mathcal A(S)$.
Thus, $(K_1,\Theta)$ determines which claims are facts, whereas $K$ determines  which facts are observed during training. 

The mutual information (MI) $I(W;S)$
is a central quantity in information-theoretic generalization analysis~\citep{xu2017information}. This quantity and its conditional variants have also been used to quantify information retention and memorization in learning~\citep{bassily2018learners,brown2021memorization}. More recently, related MI measures have been used to formalize the effective capacity of LLMs~\citep{gu2026data}. 
In our setting, $I(W;S)$ quantifies how much information the learned model retains, on average, about the realized training set and serves as a data-dependent information budget. 
The following lemma decomposes this budget into three components.

\begin{restatable}{lem}{threeMIdecomposition}\label{3MI}
In the rules-and-facts setting, we have
\begin{align}
\label{eq:rules-facts-decomposition}
I(W;S)
={}&
\underbrace{I(W;\Theta)}_{\text{Learning}}
+
\underbrace{I(W;B\mid\Theta)}_{\text{Atomic fact memory}}
+
\underbrace{I(W;K\mid\widetilde Z)}_{\text{Excess memorization}}.
\end{align}
\end{restatable}

\begin{rem}
\label{rem:all-atomic-observed}
Here, \(\widetilde Z\) contains all atomic and composite facts, and
\(S=\widetilde Z_K\) contains only the selected observations.
However, \(S\) need not contain all atomic facts, so complete atomic-fact
recovery is not guaranteed even with sufficient model capacity.
To guarantee direct access to every atomic fact, one can instead let $\widetilde{Z}=\Gamma(B,\Theta)$ contain only composite facts and define
\(S=(B,\widetilde Z_K)\), where \(K\) only  selects composite observations
for rule learning.
The exact decomposition then has the third term
\(I(W;K\mid B,\Theta)\ge I(W;K\mid\widetilde Z)\), giving
$
I(W;S)
\ge
I(W;\Theta)
+
I(W;B\mid\Theta)
+
I(W;K\mid\widetilde Z)$.
Note that the joint information-budget bound (i.e. Theorem~\ref{thm:joint-information-budget}) below also applies to this variant under the corresponding information budget, with the generalization gap evaluated over the selected and held-out composite observations.
\end{rem}
The decomposition separates three roles of model information. The first term, $I(W;\Theta)$, measures the reusable rule information retained by the model. Learning $\Theta$ from $\widetilde Z_K$ enables inference on unseen $\widetilde Z_i$ and compresses multiple related facts into shared structure.
The second term, $I(W;B\mid\Theta)$, measures atomic fact information not explained by the rules. In other words, to capture all facts, the model must store these atomic exceptions and compose them with the rules.
The third term, $I(W;K\mid\widetilde Z)$, is selector-density CMI: conditioned on the full factual supersample, it is a record not of truth, but of which facts the model happened to observe.
As above, it controls the generalization gap, and larger values may indicate sample-specific overfitting, i.e., undesirable memorization.

The first two terms jointly form the world information \(I(W;B,\Theta)\): rules compactly encode transferable structure, and applying them obtains new conclusions rather than new independent information about the world. Learning such structure, however, competes with memorizing arbitrary facts for finite capacity, echoing the rule--memory competition discussed in~\citet[Appendix~A.3]{guo2026hallucination}. Our further decomposition separates factual content from training selection, i.e. \(I(W;B,K\mid\Theta)=I(W;B\mid\Theta)+I(W;K\mid \widetilde Z)\), indicating that selection-specific memory can be discarded to promote effective learning and generalization.
This separation is closely related to the ``average excess memorization'' studied by
\cite{feldman2025trade}, who measure training data information conditional on a latent parameter drawn from a meta-distribution and specifying the data-generating distribution.
Another related distinction underlies the notion of ``unintended memorization'' in \citet{morris2026how}, who separate dataset-specific information from information about the data-generating process.
In our setting, the corresponding conditioning object is the full latent world $(B,\Theta)$. Thus, atomic-fact information belongs to the world-information
component, whereas selector information captures residual sample dependence after that world is fixed.

\begin{restatable}{thm}{jointInformationBudget}
\label{thm:joint-information-budget}
Suppose that $I(W;S)\le C$ and that
$\Theta\sim\operatorname{Unif}([N_\Theta])$, where $N_\Theta\ge2$. Assume that there exists a decoder $\widehat\Theta(W)$ satisfying $\mathbb P(\widehat\Theta\neq\Theta)\le\varepsilon_\Theta$. Define
\begin{align*}
C_{\mathrm{rule}}
\coloneqq
\log N_\Theta
-h(\varepsilon_\Theta)
-\varepsilon_\Theta\log(N_\Theta-1), \quad\text{ and }
C_{\mathrm{gen}}
\coloneqq
\frac{2nm\max\{n,m\}}{N^2}
|\mathcal E(\mathcal A)|^2. 
\end{align*}
Then $I(W;\Theta)\ge C_{\mathrm{rule}}$ and
$I(W;K\mid\widetilde Z)\ge C_{\mathrm{gen}}$.

Moreover, for a threshold decision on atomic claims, let $\alpha_\tau$ denote fact recall and $\beta_\tau$ the false-positive hallucination rate. In the sparse atomic-fact limit $\rho_1=\frac{n_1}{N_1}\to0$, perfect atomic recall, $\alpha_\tau=1$, implies
\[
\beta_\tau
\ge
\exp\left(
-\frac{C-C_{\mathrm{rule}}-C_{\mathrm{gen}}}{n_1}
\right).
\]
\end{restatable}

Theorem~\ref{thm:joint-information-budget} places rule recovery, sample dependence, and hallucination suppression under the total information budget $C$, unlike Theorem~\ref{thm:fpr-random-fact}, which assumes a score-level information budget $R$ per fact.
At perfect atomic recall, its sparse-limit constraint can be written as
$
C_{\mathrm{rule}}+C_{\mathrm{gen}}+n_1\log\frac{1}{\beta_\tau}\le C.
$
Suppressing false positives to rate \(\beta_\tau\) therefore requires at least \(\log(1/\beta_\tau)\) nats per atomic fact on average. The remaining terms $C_{\mathrm{rule}}$ and $C_{\mathrm{gen}}$ quantify the minimum information committed to rule recovery and certified by the generalization gap, respectively. At fixed \(C\), these commitments reduce the budget available for atomic factual discrimination. 
This perspective also motivates retrieval-augmented generation (RAG)~\citep{lewis2020retrieval}. As also discussed by \citet{guo2026hallucination}, external non-parametric memory can alleviate the need to encode all factual knowledge in model parameters.

\subsection{Data Mixing and the Cost of Factual Recall}
\label{sec:gaussian-data-mixing}

Theorem~\ref{thm:joint-information-budget} identifies a competition between reusable rule information and atomic-fact memory under a shared information budget. It leaves open an operational question: \emph{how does the training mixture affect the factuality behavior of a particular learner?} 
We study an additive Gaussian setting with a restricted generative model that represents both fact types through a single location parameter, allowing explicit characterization of fact recall and false-positive rates.

\begin{figure}[H]
    \centering    \includegraphics[width=0.95\linewidth]{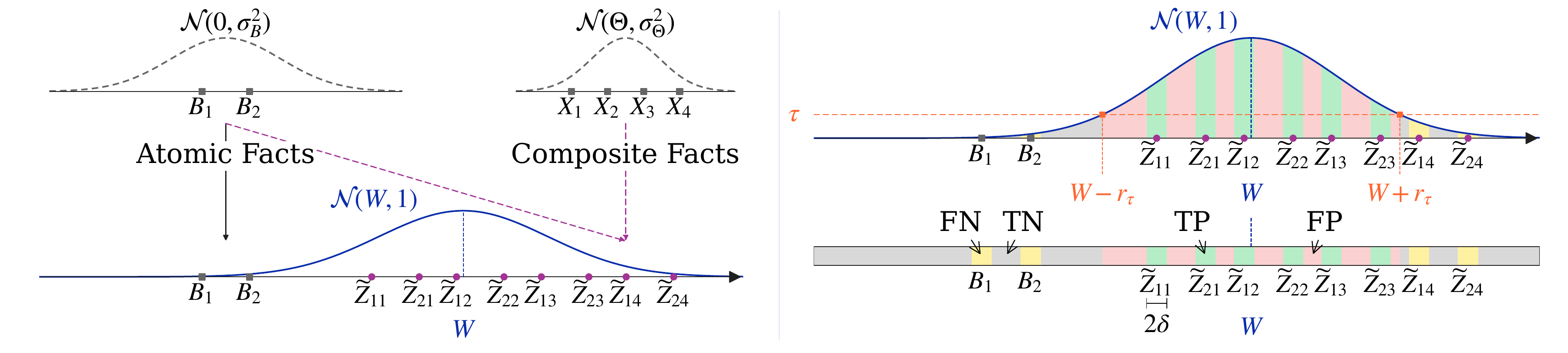}
    
    \caption{Left: atomic and composite fact generation and fitting with $\mathcal N(W,1)$. Right: factuality tolerance and the acceptance threshold jointly define classification outcomes.}
    \label{fig:gaussian_facts}
    
\end{figure}

\textbf{Atomic facts and rule-derived facts.} Let atomic facts be $B=(B_1,\ldots,B_{n_1})$ with $B_i\overset{\mathrm{i.i.d.}}{\sim}\mathcal N(0,\sigma_B^2)$, where $\sigma_B>0$. Draw $\Theta\sim P_\Theta$ on $[0,\infty)$ and, conditional on $\Theta$, generate additive rule instances
$X=(X_1,\ldots,X_{n_2})$ with $X_j\mid\Theta\overset{\mathrm{i.i.d.}}{\sim}\mathcal N(\Theta,\sigma_\Theta^2)$, where $\sigma_\Theta^2\in(0,\infty)$ is a prescribed function of $\Theta$. We assume $B\indp\Theta$.
Each rule instance is shared across atomic facts, producing $\widetilde Z\coloneqq\Gamma(B,X)\in\mathbb R^{n_1\times n_2}$ with $\widetilde Z_{ij}\coloneqq B_i+X_j$.
Thus the composite pool contains $N=n_1n_2$ facts. 
Independently of $(\Theta,B,X)$, draw $K$ uniformly over the size-$n$ subsets of $[n_1]\times[n_2]$.
As in Remark~\ref{rem:all-atomic-observed}, training uses all atomic facts and the selected composite facts, namely $S\coloneqq(B,\widetilde Z_K)$ with $|S|=n_1+n$, where $K_{ij}\coloneqq\mathbf 1\{(i,j)\in K\}$.
Define the composite proportion in the training sample by $a_N\coloneqq\frac{n}{n_1+n}$. 

\textbf{Shared generative representation.} Let $\phi$ and $\Phi$ denote the standard Gaussian density and distribution functions. 
We fit the unit-variance Gaussian density $p_w(z)=\phi(z-w)$ by minimizing the average negative log-likelihood.  In fact, this is a simple analogue of LLM pretraining, which learns a distribution over sequences.
Up to an additive constant, the per-example loss is $\ell(w,z)=\tfrac12(z-w)^2$. The empirical objective and its minimizer are therefore
\begin{align*}
\widehat L_S(w)
\coloneqq
\frac{\sum_i\ell(w,B_i)+\sum_{i,j}K_{ij}\ell(w,B_i+X_j)}{n_1+n}, \ \
\widehat W_N
\coloneqq
\frac{\sum_iB_i+\sum_{i,j}K_{ij}(B_i+X_j)}{n_1+n}.
\end{align*}
Let $W_N$ denote the learned parameter, allowing for an approximate
rather than exact minimizer.

\textbf{Recall and hallucination rate.} Let $q(W_N,z)=p_{W_N}(z)$, namely using the likelihood $p_{W_N}(z)$ as a factuality score. For $r\geq0$, set $\tau(r)=\phi(r)$, so that
$A_{\tau(r)}(z)=\mathbf 1\{p_{W_N}(z)\geq\tau(r)\}=\mathbf 1\{|z-W_N|\leq r\}$. Positive queries are sampled uniformly from the atomic or composite
fact pool. Their conditional acceptance probabilities are
\begin{align*}
    \alpha_N^{\mathrm{ato}}(r)
    &\coloneqq
    \frac1{n_1}
    \sum_{i=1}^{n_1}
    \mathbf 1\{|B_i-W_N|\leq r\}, \quad
    \alpha_N^{\mathrm{comp}}(r)
    \coloneqq
    \frac1N
    \sum_{i=1}^{n_1}\sum_{j=1}^{n_2}
    \mathbf 1\{|\widetilde Z_{ij}-W_N|\leq r\}.
\end{align*}
The composite recall is evaluated on the full composite pool,
including facts not selected for training.

To define nonfact queries in this continuous setting, introduce a
factuality tolerance $\delta_N>0$ and let $\mathcal T_N\coloneqq\bigcup_{i=1}^{n_1}[B_i-\delta_N,B_i+\delta_N]\;\cup\;\bigcup_{i=1}^{n_1}\bigcup_{j=1}^{n_2}[\widetilde Z_{ij}-\delta_N,\widetilde Z_{ij}+\delta_N]$.
Conditional on the realized $(B,X,W_N)$, draw a fresh candidate
$Z\sim\mathcal N(W_N,1)$ and define $\beta_N(r)\coloneqq\mathbb P\bigl(|Z-W_N|\leq r
        \,\big|\,Z\notin\mathcal T_N\bigr)$.
Thus $\beta_N(r)$ is the false-positive (hallucination) rate, i.e., the probability of accepting a model-proposed candidate conditioned on it being a nonfact.

\begin{restatable}{thm}{gaussianRecallLimit}
\label{thm:gaussian-recall}
Fix $\theta>0$ and condition on $\Theta=\theta$. Let $\sigma_{\mathrm c}\coloneqq\sqrt{\sigma_B^2+\sigma_\theta^2}$. Suppose that
$n_1,n_2,n\to\infty$, $a_N\to a\in[0,1]$,
$N\delta_N\to0$, and that $W_N-\widehat W_N\xrightarrow{\mathbb P}0$.
Then $W_N\xrightarrow{\mathbb P}a\theta$.
Moreover, for every fixed $r\geq0$,
\begin{align}
    \alpha_N^{\mathrm{ato}}(r)
    &\xrightarrow{\mathbb P}
    \alpha^{\mathrm{ato}}(r;a)
    \coloneqq
    \Phi\!\left(\frac{r-a\theta}{\sigma_B}\right)
    +
    \Phi\!\left(\frac{r+a\theta}{\sigma_B}\right)-1,
    \label{eq:gaussian-atomic-limit}\\
    \alpha_N^{\mathrm{comp}}(r)
    &\xrightarrow{\mathbb P}
    \alpha^{\mathrm{comp}}(r;a)\coloneqq
    \Phi\!\left(\frac{r-(1-a)\theta}{\sigma_{\mathrm c}}\right)
    +
    \Phi\!\left(\frac{r+(1-a)\theta}{\sigma_{\mathrm c}}\right)-1,
    \label{eq:gaussian-composite-limit}\\
    \beta_N(r)
    &\xrightarrow{\mathbb P}
    \beta(r)\coloneqq2\Phi(r)-1.
    \label{eq:gaussian-fpr-limit}
\end{align}

\end{restatable}

Notably, at any fixed radius $r>0$, increasing the composite proportion $a$ moves the learned center $a\theta$ from the atomic population toward the composite population. Atomic recall therefore decreases, whereas composite recall increases. Importantly, the limiting false-positive rate at a fixed radius is independent of $a$, hence the mixture affects false positives only through the radius required to maintain the desired recall. This motivates choosing the mixture jointly with the acceptance threshold.

\begin{restatable}
{cor}{optimalMixtureRecall}
\label{cor:optimal-mixture-recall}
Under the conditions of Theorem~\ref{thm:gaussian-recall},
fix a target recall $\alpha_0\in(1/2,1)$. For $a\in[0,1]$, define $r_g(a)\coloneqq\inf\{r\geq0:\alpha^g(r;a)\geq\alpha_0\}$ for $g\in\{\mathrm{ato},\mathrm{comp}\}$, the smallest radius at which the limiting recall of population $g$ reaches $\alpha_0$.
Let $r_{\alpha_0}(a)\coloneqq\max\{r_{\mathrm{ato}}(a),r_{\mathrm{comp}}(a)\}$,
and let $\beta_{\alpha_0}(a)\coloneqq\beta(r_{\alpha_0}(a))$ be its limiting false-positive rate.
Then
\begin{equation}
\label{eq:gaussian-optimal-fpr-lower}
    \min_{a\in[0,1]}\beta_{\alpha_0}(a)
    \geq
    2\Phi\!\left(
        \frac{\theta}{2}
        +
        \frac{\sigma_B+\sigma_{\mathrm c}}{2}
        \Phi^{-1}(\alpha_0)
    \right)-1.
\end{equation}
\end{restatable}

\begin{rem}
    If $r_{\mathrm{ato}}(1)>r_{\mathrm{comp}}(1)$, there is a unique
$a_*\in(0,1)$ satisfying $r_{\mathrm{ato}}(a_*)=r_{\mathrm{comp}}(a_*)$. In this case, $\beta_{\alpha_0}(a)$ is strictly decreasing on $[0,a_*]$ and strictly increasing on $[a_*,1]$. Otherwise, its unique minimizer is $a_*=1$.
Moreover, $\sigma_{\mathrm c}>\sigma_B$ and $\alpha_0>1/2$ imply $r_{\mathrm{comp}}(1/2)>r_{\mathrm{ato}}(1/2)$, and hence $a_*>1/2$. Thus, equal recall targets do not imply equal training proportions, rather, the optimal mixture favors composite facts to compensate for their greater dispersion.
\end{rem}

Corollary~\ref{cor:optimal-mixture-recall} shows that an interior optimal mixture balances the acceptance radii required by the two fact populations, rather than their sample counts. An overly atomic-heavy mixture requires a wide acceptance region
to retain composite facts, and an overly composite-heavy mixture does the opposite. Mixture selection can reposition the shared center, but cannot remove the separation or dispersion of the two populations. Consequently, even the optimal mixture incurs a false-positive cost when both recall requirements must be met. For fixed variances, greater separation $\theta$ raises the lower bound in~\eqnref{eq:gaussian-optimal-fpr-lower}; increasing either standard deviation or the recall target also raises this bound.

\section{Experiments}
\label{sec:experiments-relational-design}
In the preceding analysis, we lower-bound atomic hallucination under a shared information budget, revealing competition between rule learning and atomic-fact memory (Theorem~\ref{thm:joint-information-budget}). The Gaussian analysis further shows that the cost of recalling both fact types depends on the training mixture and the underlying distributional structure (Corollary~\ref{cor:optimal-mixture-recall}). Motivated by these results, we design a controlled relational task in the rules-and-facts setting, where composite-fact exposure adjusts the training mixture and rule count varies the structural learning burden. 
Across model sizes, we examine how these two factors interact with capacity to shape atomic hallucination.

\begin{figure}[!ht]
    \centering
    \includegraphics[width=0.96\linewidth]{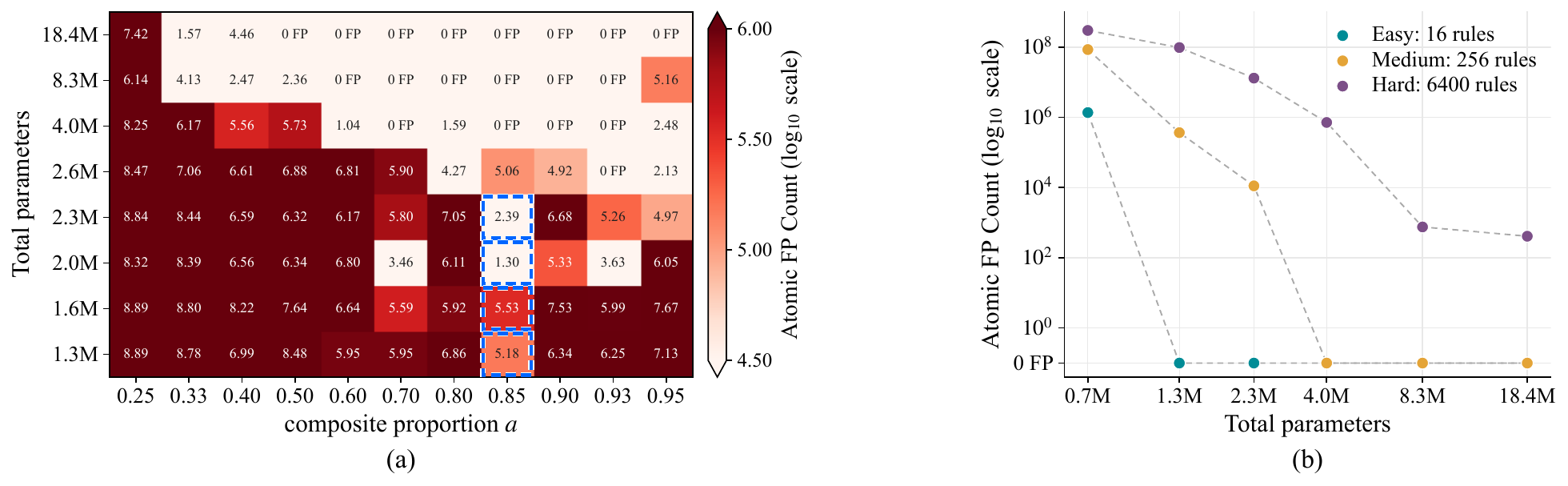}
    
    \caption{AtomicFP at perfect fact recall across model sizes: (a) varying composite training proportion $a$; (b) varying rule count at fixed $a=0.5$.}
    \label{fig:atomic-fpr-g3}
    
\end{figure}

\textbf{Task setup.}
We construct a controlled rules-and-facts knowledge graph, inspired
by~\citet{wang2024grokked,yao2026compositional}, to test whether
a model can retain atomic facts and reuse composition rules.
Entities are nodes, and each fact $z=(e_h,r_i,e_t)$ is a directed edge.
Under the rule $r_a\circ r_b=r_c$, the atomic facts $(e_1,r_a,e_2)$ and
$(e_2,r_b,e_3)$ imply the composite fact $(e_1,r_c,e_3)$.
In a real-world
analogue, Alice's father is Bob and Bob's mother is Carol, so Carol is
Alice's grandmother.

\textbf{Training and evaluation.}
We train randomly initialized Transformers autoregressively with varying parameter counts 
on symbolic three-token sequences $(e_h,r_i,e_t)$, mirroring
language-model pretraining. 
The training set includes all $n_1$ atomic facts in
$B$ and $n$ composite facts from $\widetilde Z_{\mathrm{comp}}$.
We fix $n_1$ and vary $n$ to adjust the composite proportion
$a=n/(n_1+n)$ over $[0.25,0.95]$.
To evaluate hallucination for a trained model $w$, we use the predicted
tail probability $q(w,z)\coloneqq p_w(e_t\mid e_h,r_i)$ as the
factuality score.
We construct the negative set
$B^-
\coloneqq
\bigl\{(e_h,r_i,e):
(e_h,r_i,e_t)\in B,\ e\neq e_t\bigr\}$ and report
$
\mathrm{AtomicFP}
\coloneqq
\sum_{z\in B^-}
\mathbf 1\{q(w,z)\geq\tau_*\}$,
where $\tau_*\coloneqq
\min_{z\in B\cup\widetilde Z_{\mathrm{comp}}}q(w,z)$.
With $|B|$ and entity vocabulary fixed, $|B^-|$ is constant, so $\mathrm{AtomicFP}$ is proportional
to the empirical atomic hallucination rate at perfect fact recall.
We also fix $n_1$ and $a$ and vary the number of rules across
$16$, $256$, and $6400$ to examine how the rule-learning burden
changes model capacity requirements.

\textbf{Results and analysis.} Figure~\ref{fig:atomic-fpr-g3}(a) shows that small and medium-sized models generally require intermediate mixtures to achieve low AtomicFP,  while larger models are less sensitive to this choice. 
Thus, additional composite training facts need not improve atomic factual discrimination.
Meanwhile, Figure~\ref{fig:atomic-fpr-g3}(b) shows that more composite rules increase hallucination at a fixed model size.
Although the atomic and composite training counts remain unchanged, richer rule structure demands more bits to encode the factual world, thereby imposing a tighter capacity-per-unit-fact constraint.
Overall, greater rule-related exposure or higher rule complexity can incur a cost in atomic factual discrimination.
Scaling model capacity helps ease this trade-off, consistent with  our theory in the rules-and-facts setting. We provide more experimental details and results in Appendix~\ref{appendix-exp}.

\section{Conclusion and Limitations}

In this work, we introduce a selector-density CMI framework that unifies generalization, memorization, and hallucination, and distinguish the factual memory needed for reliable discrimination from the training-selection information associated with sample dependence.
Our information-theoretic analysis uses uniformly selected fact sets and a specified separation between atomic facts and independent latent rules. Hallucination is measured through threshold-based false positives, and the joint budget bound assumes perfect atomic recall. The bounds characterize information requirements but do not determine how training dynamics allocate capacity or whether the bounds are attainable. The Gaussian model restricts both fact types to a single shared location parameter, while the synthetic relational experiments provide qualitative evidence without directly measuring the information components. Extending the joint analysis to partial recall and more general fact distributions, and connecting information allocation to training dynamics and factuality in open-ended language generation, remain important directions for future work.

\clearpage
\bibliographystyle{plainnat}
\bibliography{\arxivroot ref}

\begin{thebibliography}{35}
\providecommand{\natexlab}[1]{#1}
\providecommand{\url}[1]{\texttt{#1}}
\expandafter\ifx\csname urlstyle\endcsname\relax
  \providecommand{\doi}[1]{doi: #1}\else
  \providecommand{\doi}{doi: \begingroup \urlstyle{rm}\Url}\fi

\bibitem[Bassily et~al.(2018)Bassily, Moran, Nachum, Shafer, and Yehudayoff]{bassily2018learners}
Raef Bassily, Shay Moran, Ido Nachum, Jonathan Shafer, and Amir Yehudayoff.
\newblock Learners that use little information.
\newblock In \emph{Algorithmic Learning Theory}. PMLR, 2018.

\bibitem[Brown et~al.(2021)Brown, Bun, Feldman, Smith, and Talwar]{brown2021memorization}
Gavin Brown, Mark Bun, Vitaly Feldman, Adam Smith, and Kunal Talwar.
\newblock When is memorization of irrelevant training data necessary for high-accuracy learning?
\newblock In \emph{Proceedings of the 53rd annual ACM SIGACT symposium on theory of computing}, pages 123--132, 2021.

\bibitem[Brown et~al.(2020)Brown, Mann, Ryder, Subbiah, Kaplan, Dhariwal, Neelakantan, Shyam, Sastry, Askell, Agarwal, Herbert-Voss, Krueger, Henighan, Child, Ramesh, Ziegler, Wu, Winter, Hesse, Chen, Sigler, Litwin, Gray, Chess, Clark, Berner, McCandlish, Radford, Sutskever, and Amodei]{brown2020language}
Tom~B. Brown, Benjamin Mann, Nick Ryder, Melanie Subbiah, Jared~D. Kaplan, Prafulla Dhariwal, Arvind Neelakantan, Pranav Shyam, Girish Sastry, Amanda Askell, Sandhini Agarwal, Ariel Herbert-Voss, Gretchen Krueger, Tom Henighan, Rewon Child, Aditya Ramesh, Daniel~M. Ziegler, Jeffrey Wu, Clemens Winter, Christopher Hesse, Mark Chen, Eric Sigler, Mateusz Litwin, Scott Gray, Benjamin Chess, Jack Clark, Christopher Berner, Sam McCandlish, Alec Radford, Ilya Sutskever, and Dario Amodei.
\newblock Language models are few-shot learners.
\newblock In \emph{Advances in Neural Information Processing Systems}, volume~33, 2020.

\bibitem[Dong et~al.(2025)Dong, Guo, Gong, Wen, and Li]{dong2025exactly}
Yuxin Dong, Haoran Guo, Tieliang Gong, Wen Wen, and Chen Li.
\newblock Exactly tight information-theoretic generalization bounds via binary jensen-shannon divergence.
\newblock In \emph{Forty-second International Conference on Machine Learning}, 2025.

\bibitem[Donsker and Varadhan(1983)]{donsker1983asymptotic}
Monroe~D Donsker and SR~Srinivasa Varadhan.
\newblock Asymptotic evaluation of certain markov process expectations for large time. iv.
\newblock \emph{Communications on pure and applied mathematics}, 36\penalty0 (2):\penalty0 183--212, 1983.

\bibitem[Feldman(2020)]{feldman2020does}
Vitaly Feldman.
\newblock Does learning require memorization? a short tale about a long tail.
\newblock In \emph{Proceedings of the 52nd annual ACM SIGACT symposium on theory of computing}, pages 954--959, 2020.

\bibitem[Feldman et~al.(2025)Feldman, Kornowski, and Lyu]{feldman2025trade}
Vitaly Feldman, Guy Kornowski, and Xin Lyu.
\newblock Trade-offs in data memorization via strong data processing inequalities.
\newblock In \emph{The Thirty Eighth Annual Conference on Learning Theory}, volume 291, pages 1935--1973. PMLR, 2025.

\bibitem[Garg et~al.(2022)Garg, Tsipras, Liang, and Valiant]{garg2022what}
Shivam Garg, Dimitris Tsipras, Percy Liang, and Gregory Valiant.
\newblock What can transformers learn in-context? a case study of simple function classes.
\newblock In \emph{Advances in Neural Information Processing Systems}, volume~35, 2022.

\bibitem[Gu et~al.(2025)Gu, Lyu, Li, and Zhang]{gu2026data}
Xinran Gu, Kaifeng Lyu, Jiazheng Li, and Jingzhao Zhang.
\newblock Data mixing can induce phase transitions in knowledge acquisition.
\newblock In \emph{Advances in Neural Information Processing Systems}, volume~38, pages 163200--163244, 2025.
\newblock URL \url{https://proceedings.neurips.cc/paper_files/paper/2025/hash/eee7ae5cf0c4356c2aeca400771791aa-Abstract-Conference.html}.

\bibitem[Guo and Li(2026)]{guo2026hallucination}
Anxin Guo and Jingwei Li.
\newblock Hallucination is a consequence of space-optimality: A rate-distortion theorem for membership testing.
\newblock In \emph{Forty-third International Conference on Machine Learning}, 2026.
\newblock URL \url{https://openreview.net/forum?id=Aj7lNFaY0M}.

\bibitem[Haghifam et~al.(2022)Haghifam, Moran, Roy, and Dziugaite]{haghifam2022understanding}
Mahdi Haghifam, Shay Moran, Daniel~M Roy, and Gintare~Karolina Dziugaite.
\newblock Understanding generalization via leave-one-out conditional mutual information.
\newblock In \emph{2022 IEEE International Symposium on Information Theory (ISIT)}, pages 2487--2492. IEEE, 2022.

\bibitem[Harutyunyan et~al.(2021)Harutyunyan, Raginsky, Steeg, and Galstyan]{harutyunyan2021informationtheoretic}
Hrayr Harutyunyan, Maxim Raginsky, Greg~Ver Steeg, and Aram Galstyan.
\newblock Information-theoretic generalization bounds for black-box learning algorithms.
\newblock In \emph{Advances in Neural Information Processing Systems}, 2021.

\bibitem[Hellstr{\"o}m and Durisi(2022)]{hellstrom2022a}
Fredrik Hellstr{\"o}m and Giuseppe Durisi.
\newblock A new family of generalization bounds using samplewise evaluated {CMI}.
\newblock In \emph{Advances in Neural Information Processing Systems}, 2022.

\bibitem[Hellstr{\"o}m et~al.(2025)Hellstr{\"o}m, Durisi, Guedj, and Raginsky]{hellstrom2025generalization}
Fredrik Hellstr{\"o}m, Giuseppe Durisi, Benjamin Guedj, and Maxim Raginsky.
\newblock Generalization bounds: Perspectives from information theory and pac-bayes.
\newblock \emph{Foundations and Trends in Machine Learning}, 18\penalty0 (1):\penalty0 1--223, 2025.
\newblock \doi{10.1561/2200000112}.
\newblock URL \url{https://doi.org/10.1561/2200000112}.

\bibitem[Hoeffding(1963)]{hoeffding1963probability}
Wassily Hoeffding.
\newblock Probability inequalities for sums of bounded random variables.
\newblock \emph{Journal of the American Statistical Association}, 58\penalty0 (301):\penalty0 13--30, 1963.

\bibitem[Kalai and Vempala(2024)]{kalai2024calibrated}
Adam~Tauman Kalai and Santosh~S Vempala.
\newblock Calibrated language models must hallucinate.
\newblock In \emph{Proceedings of the 56th Annual ACM Symposium on Theory of Computing}, pages 160--171, 2024.

\bibitem[Kalai et~al.(2025)Kalai, Nachum, Vempala, and Zhang]{kalai2025language}
Adam~Tauman Kalai, Ofir Nachum, Santosh~S Vempala, and Edwin Zhang.
\newblock Why language models hallucinate, 2025.
\newblock URL \url{https://arxiv.org/abs/2509.04664}.

\bibitem[Kalai et~al.(2026)Kalai, Nachum, Vempala, and Zhang]{kalai2026evaluating}
Adam~Tauman Kalai, Ofir Nachum, Santosh~S Vempala, and Edwin Zhang.
\newblock Evaluating large language models for accuracy incentivizes hallucinations.
\newblock \emph{Nature}, 653:\penalty0 1047--1051, 2026.
\newblock \doi{10.1038/s41586-026-10549-w}.
\newblock URL \url{https://www.nature.com/articles/s41586-026-10549-w}.

\bibitem[Kalavasis et~al.(2025)Kalavasis, Mehrotra, and Velegkas]{kalavasis2025limits}
Alkis Kalavasis, Anay Mehrotra, and Grigoris Velegkas.
\newblock On the limits of language generation: Trade-offs between hallucination and mode collapse.
\newblock In \emph{Proceedings of the 57th Annual ACM Symposium on Theory of Computing}, pages 1732--1743, 2025.
\newblock \doi{10.1145/3717823.3718108}.
\newblock URL \url{https://arxiv.org/abs/2411.09642}.

\bibitem[Lewis et~al.(2020)Lewis, Perez, Piktus, Petroni, Karpukhin, Goyal, K{\"u}ttler, Lewis, Yih, Rockt{\"a}schel, et~al.]{lewis2020retrieval}
Patrick Lewis, Ethan Perez, Aleksandra Piktus, Fabio Petroni, Vladimir Karpukhin, Naman Goyal, Heinrich K{\"u}ttler, Mike Lewis, Wen-tau Yih, Tim Rockt{\"a}schel, et~al.
\newblock Retrieval-augmented generation for knowledge-intensive nlp tasks.
\newblock \emph{Advances in neural information processing systems}, 33:\penalty0 9459--9474, 2020.

\bibitem[Lu et~al.(2026)Lu, Frey, Kuijper, and Zhu]{lu2026unified}
Yang Lu, Matthias Frey, Margreta Kuijper, and Jingge Zhu.
\newblock On unified and sharpened cmi bounds for generalization errors, 2026.
\newblock URL \url{https://arxiv.org/abs/2605.21056}.

\bibitem[Morris et~al.(2026)Morris, Sitawarin, Guo, Kokhlikyan, Suh, Rush, Chaudhuri, and Mahloujifar]{morris2026how}
John~Xavier Morris, Chawin Sitawarin, Chuan Guo, Narine Kokhlikyan, G.~Edward Suh, Alexander~M Rush, Kamalika Chaudhuri, and Saeed Mahloujifar.
\newblock How much can language models memorize?
\newblock In \emph{Forty-third International Conference on Machine Learning}, 2026.
\newblock URL \url{https://openreview.net/forum?id=bA6BgSbaUi}.

\bibitem[Pan et~al.(2025)Pan, Wang, Liao, and Li]{pan2025understanding}
Zhixuan Pan, Shaowen Wang, Pengfei Liao, and Jian Li.
\newblock Understanding {LLM} behaviors via compression: Data generation, knowledge acquisition and scaling laws.
\newblock In \emph{Advances in Neural Information Processing Systems}, volume~38, pages 186604--186654, 2025.
\newblock \doi{10.52202/085713-5613}.
\newblock URL \url{https://proceedings.neurips.cc/paper_files/paper/2025/hash/f61d7778e89b9221d1ea0ce8428b7014-Abstract-Conference.html}.

\bibitem[Polyanskiy and Wu(2023)]{yury2023information}
Yury Polyanskiy and Yihong Wu.
\newblock \emph{Information Theory: From Coding to Learning}.
\newblock Cambridge university press, 2023.

\bibitem[Radford et~al.(2019)Radford, Wu, Child, Luan, Amodei, and Sutskever]{radford2019language}
Alec Radford, Jeffrey Wu, Rewon Child, David Luan, Dario Amodei, and Ilya Sutskever.
\newblock Language models are unsupervised multitask learners.
\newblock Technical report, OpenAI, 2019.

\bibitem[Rammal et~al.(2022)Rammal, Achille, Golatkar, Diggavi, and Soatto]{rammal2022on}
Mohamad~Rida Rammal, Alessandro Achille, Aditya Golatkar, Suhas Diggavi, and Stefano Soatto.
\newblock On leave-one-out conditional mutual information for generalization.
\newblock In \emph{Advances in Neural Information Processing Systems}, 2022.

\bibitem[Sanh et~al.(2022)Sanh, Webson, Raffel, Bach, Sutawika, Alyafeai, Chaffin, Stiegler, Raja, Dey, Bari, Xu, Thakker, Sharma, Szczechla, Kim, Chhablani, Nayak, Datta, Chang, Jiang, Wang, Manica, Shen, Yong, Pandey, Bawden, Wang, Neeraj, Rozen, Sharma, Santilli, Fevry, Fries, Teehan, Le~Scao, Biderman, Gao, Wolf, and Rush]{sanh2022multitask}
Victor Sanh, Albert Webson, Colin Raffel, Stephen Bach, Lintang Sutawika, Zaid Alyafeai, Antoine Chaffin, Arnaud Stiegler, Arun Raja, Manan Dey, M.~Saiful Bari, Canwen Xu, Urmish Thakker, Shanya Sharma, Eliza Szczechla, Taewoon Kim, Gunjan Chhablani, Nihal Nayak, Debajyoti Datta, Jonathan Chang, Mike Tian-Jian Jiang, Han Wang, Matteo Manica, Sheng Shen, Zheng~Xin Yong, Harshit Pandey, Rachel Bawden, Thomas Wang, Trishala Neeraj, Jos Rozen, Abheesht Sharma, Andrea Santilli, Thibault Fevry, Jason~Alan Fries, Ryan Teehan, Teven Le~Scao, Stella Biderman, Leo Gao, Thomas Wolf, and Alexander~M. Rush.
\newblock Multitask prompted training enables zero-shot task generalization.
\newblock In \emph{International Conference on Learning Representations}, 2022.

\bibitem[Steinke and Zakynthinou(2020)]{steinke2020reasoning}
Thomas Steinke and Lydia Zakynthinou.
\newblock Reasoning about generalization via conditional mutual information.
\newblock In \emph{Proceedings of Thirty Third Conference on Learning Theory}, volume 125 of \emph{Proceedings of Machine Learning Research}, pages 3437--3452. PMLR, 2020.
\newblock URL \url{https://proceedings.mlr.press/v125/steinke20a.html}.

\bibitem[Wang et~al.(2024)Wang, Yue, Su, and Sun]{wang2024grokked}
Boshi Wang, Xiang Yue, Yu~Su, and Huan Sun.
\newblock Grokking of implicit reasoning in transformers: A mechanistic journey to the edge of generalization.
\newblock In \emph{Advances in Neural Information Processing Systems}, volume~37, pages 95238--95265, 2024.
\newblock URL \url{https://proceedings.neurips.cc/paper_files/paper/2024/hash/ad217e0c7fecc71bdf48660ad6714b07-Abstract-Conference.html}.

\bibitem[Wang and Mao(2023{\natexlab{a}})]{Wang2023sampleconditioned}
Ziqiao Wang and Yongyi Mao.
\newblock Sample-conditioned hypothesis stability sharpens information-theoretic generalization bounds.
\newblock In \emph{Thirty-seventh Conference on Neural Information Processing Systems}, 2023{\natexlab{a}}.

\bibitem[Wang and Mao(2023{\natexlab{b}})]{wang2023tighter}
Ziqiao Wang and Yongyi Mao.
\newblock Tighter information-theoretic generalization bounds from supersamples.
\newblock In \emph{International Conference on Machine Learning}. PMLR, 2023{\natexlab{b}}.

\bibitem[Wang and Mao(2024)]{wang2024generalization}
Ziqiao Wang and Yongyi Mao.
\newblock Generalization bounds via conditional {$f$}-information.
\newblock In \emph{The Thirty-eighth Annual Conference on Neural Information Processing Systems}, 2024.

\bibitem[Xu and Raginsky(2017)]{xu2017information}
Aolin Xu and Maxim Raginsky.
\newblock Information-theoretic analysis of generalization capability of learning algorithms.
\newblock \emph{Advances in Neural Information Processing Systems}, 2017.

\bibitem[Yao et~al.(2026)Yao, Ren, Liao, Ding, and Liu]{yao2026compositional}
Xinhao Yao, Ruifeng Ren, Yun Liao, Lizhong Ding, and Yong Liu.
\newblock Compositional generalization from learned skills via {CoT} training: A theoretical and structural analysis for reasoning.
\newblock In \emph{The Fourteenth International Conference on Learning Representations}, 2026.
\newblock URL \url{https://openreview.net/forum?id=VLjTqLB0J9}.

\bibitem[Ye et~al.(2026)Ye, Feldman, and Talwar]{ye2026cram}
Jiayuan Ye, Vitaly Feldman, and Kunal Talwar.
\newblock Cram less to fit more: Training data pruning improves memorization of facts, 2026.
\newblock URL \url{https://arxiv.org/abs/2604.08519}.

\end{thebibliography}

\newpage

\appendix
\startcontents[appendix]
\section*{Appendix Contents}
\printcontents[appendix]{}{1}{\setcounter{tocdepth}{2}}
\clearpage
{\LARGE \centering \textbf{Appendix} \par}

\section{Related works}
\label{sec:related-work}

\subsection{CMI Bounds for Generalization}

Conditional mutual information provides a well-established framework for controlling generalization through information about training membership~\citep{steinke2020reasoning,hellstrom2025generalization}.
Conditioning on a supersample isolates the information a model reveals about the training selection, connecting generalization to compression, stability, and privacy.
Subsequent work measures this information in predictions through functional CMI (f-CMI)~\citep{harutyunyan2021informationtheoretic}, in losses through evaluated CMI (e-CMI)~\citep{hellstrom2022a}, and in loss differences through ld-CMI~\citep{wang2023tighter}.
These refinements discard information irrelevant to the evaluated loss and can sharpen parameter-level bounds.
Further improvements use sample-conditioned hypothesis stability~\citep{Wang2023sampleconditioned}, conditional $f$-information based on $f$-divergences~\citep{wang2024generalization}, and binary Jensen--Shannon divergence~\citep{dong2025exactly}.

The selection mechanism also matters.
Leave-one-out CMI connects generalization to cross-validation and stability~\citep{haghifam2022understanding,rammal2022on}; the leave-$m$-out framework unifies a broader family of supersample constructions~\citep{lu2026unified}.
Our uniform selection of $n$ samples from $N$ candidates is its unpartitioned case.
We combine this variable selection density with evaluator-level information: losses yield generalization bounds, whereas factuality scores yield limits on factual discrimination.
The common object is the information revealed about selection; its statistical role depends on what the selector and evaluator represent.

\subsection{Theoretical Perspectives on Hallucination}

Theoretical accounts of hallucination use several complementary frameworks: calibration and Good--Turing missing-mass estimation~\citep{kalai2024calibrated}, Is-It-Valid classification and evaluation incentives~\citep{kalai2025language,kalai2026evaluating}, and rate--distortion theory for membership testing~\citep{guo2026hallucination}.
Compression-based analyses explain why limited capacity can prioritize frequent knowledge over rare facts~\citep{pan2025understanding}; formal language generation exposes a trade-off between validity and coverage~\citep{kalavasis2025limits}.
The membership-testing formulation is closest to our random-fact analysis: in the sparse-fact limit, the minimum KL divergence between admissible fact and nonfact score laws characterizes the optimal memory cost~\citep{guo2026hallucination}.
Our Selector-Density CMI framework provides a unified view of factual discrimination and membership information.

\subsection{Memorization and Capacity Allocation}

Memorization can be necessary for generalization.
\citet{feldman2020does} shows that near-optimal prediction on long-tailed subpopulation mixtures can require memorizing rarely observed labels.
This connects to singleton-based hallucination bounds~\citep{kalai2024calibrated}: rare observations create a need for memorization in the former and reveal missing factual evidence in the latter.
\citet{brown2021memorization} further construct learning problems where high accuracy requires storing substantial training information, including task-irrelevant details.
Using strong data processing inequalities, \citet{feldman2025trade} characterize how the excess memorization required for accurate learning decreases as more samples become available.
Their meta-distribution formulation conditions on the data-generating parameter, separating distributional knowledge from residual sample information.
\citet{morris2026how} make a related distinction through compression-based measures of distributional knowledge and sample-specific memorization.

For factual knowledge, \citet{ye2026cram} measure information about world facts, relate it to factual accuracy, and show that data pruning can improve the use of limited memorization capacity.
\citet{gu2026data} show that mixing random facts with a structurally complex data source can reallocate finite capacity and induce phase transitions in knowledge acquisition.
The mixing ratio changes the loss reduction available per unit of capacity, so the optimal allocation across sources can shift abruptly.
In our decomposition, world information is resolved into transferable rules and atomic facts, while selector CMI measures training-choice information conditional on the factual world.

\section{Calibration at Finite Truth Density}
\label{sec:appendix-calibration}
As a by-product result, we next consider perfectly calibrated scores, without requiring the
sparse limit. Let $T\in[0,1]$ be a reported probability of truth.
Perfect calibration means
$\mathbb{P}(U=1\mid T)=T$ almost surely. Since $U$ is binary, this condition is equivalently written as $\mathbb{E}[U\mid T]=T$. We have the following bounds.

\begin{restatable}{cor}{calibratedROC}
\label{cor:calibrated-roc}
Under perfect calibration, for $\tau\in(0,1]$ with $\mathbb{P}(T\geq\tau)>0$,
the precision of the accepted claims satisfies $\mathbb{P}(U=1\mid T\geq\tau)
\geq\tau$,
and therefore
\[
\br{\alpha_{\tau}-\sqrt{\frac{h(\rho)-\mathbb{E}_{T}[h(T)]}{2\rho(1-\rho)}}}_{+}\leq\beta_\tau\leq\frac{\rho(1-\tau)}{(1-\rho)\tau}\alpha_\tau.
\]
\end{restatable}

\begin{proof}
We first establish the change-of-measure identities used below.
For any measurable set $A\subseteq[0,1]$, calibration gives
\begin{align*}
\rho P(A)
&=\mathbb E[\mathbf 1\{T\in A\}U]
 =\mathbb E[\mathbf 1\{T\in A\}\mathbb E[U\mid T]]
 =\int_A t\,M_\rho(dt),\\
(1-\rho)Q(A)
&=\mathbb E[\mathbf 1\{T\in A\}(1-U)]
 =\mathbb E[\mathbf 1\{T\in A\}(1-T)]
 =\int_A(1-t)\,M_\rho(dt).
\end{align*}
Thus $\frac{dP}{dM_\rho}(t)=t/\rho$ and
$\frac{dQ}{dM_\rho}(t)=(1-t)/(1-\rho)$, so
\[
\alpha_\tau=\frac1\rho\int_{[\tau,1]}t\,M_\rho(dt),
\qquad
\beta_\tau=\frac1{1-\rho}\int_{[\tau,1]}(1-t)\,M_\rho(dt).
\]
Perfect calibration and the tower property give
\begin{align*}
\mathbb P(U=1\mid T\geq\tau)
&=\mathbb E[U\mid T\geq\tau]\\
&=\mathbb E\!\left[\mathbb E[U\mid T]\mid T\geq\tau\right]\\
&=\mathbb E[T\mid T\geq\tau]\geq\tau.
\end{align*}
Bayes' rule then yields
\[
\mathbb P(U=1\mid T\geq\tau)
=\frac{\rho\alpha_\tau}
{\rho\alpha_\tau+(1-\rho)\beta_\tau}\geq\tau,
\]
which gives the stated upper bound on $\beta_\tau$.

For the lower bound, perfect calibration implies
$U\mid T\sim\Bern(T)$. Therefore,
\[
I(U;T)=H(U)-H(U\mid T)
=h(\rho)-\mathbb E_T[h(T)].
\]
We also need the weighted Pinsker inequality. For
$M_\rho=\rho P+(1-\rho)Q$, Pinsker's inequality gives
\begin{align*}
\TV(P,M_\rho)&=(1-\rho)\TV(P,Q)
 \leq\sqrt{\frac{\kl{P}{M_\rho}}2},\\
\TV(Q,M_\rho)&=\rho\TV(P,Q)
 \leq\sqrt{\frac{\kl{Q}{M_\rho}}2},
\end{align*}
where $\TV$ denotes total variation distance. Multiplying the squared
inequalities by $2\rho$ and $2(1-\rho)$, respectively, and adding gives
\[
J_\rho(P,Q)\geq
2\rho(1-\rho)^2\TV^2(P,Q)
+2(1-\rho)\rho^2\TV^2(P,Q)
=2\rho(1-\rho)\TV^2(P,Q).
\]
For $A_\tau=\mathbf 1\{T\geq\tau\}$, data processing therefore implies
\[
I(U;T)\geq I(U;A_\tau)
=J_\rho\bigl(\Bern(\alpha_\tau),\Bern(\beta_\tau)\bigr)
\geq2\rho(1-\rho)(\alpha_\tau-\beta_\tau)^2.
\]
Combining this with the entropy identity yields
\[
|\alpha_\tau-\beta_\tau|
\leq\sqrt{\frac{h(\rho)-\mathbb E_T[h(T)]}{2\rho(1-\rho)}}.
\]
The stated lower bound follows because $\beta_\tau\geq0$.
\end{proof}

The two sides of the bound above serve different purposes.
The lower bound quantifies how limited truth-selector information
restricts the separation between recall and hallucination rate.
The upper bound follows from the precision guarantee associated with
calibration, but does not guarantee substantial recall.
Our result is complementary to \citet{kalai2024calibrated}, who establish statistical hallucination lower bounds under a semantic generative-calibration condition and assumptions on arbitrary facts whose truth cannot be inferred from
the observed data. Here, we instead study binary correctness calibration,
$\mathbb{P}(U=1\mid T)=T$, and the false-positive rate of a
claim-acceptance rule. These differences in calibration and error definitions preclude a direct identification of the two results.
In particular, the calibrated oracle $T=U$ achieves zero false
positives and perfect recall. Note that its existence as a joint distribution does not imply that it is learnable under the information constraints of a given training problem.

\section{Experimental Details and Additional Results}
\label{appendix-exp}

\subsection{Data Construction and Splits}

We build the knowledge graph around centers $e_c$, each connecting
$d_{\mathrm{in}}$ incoming and $d_{\mathrm{out}}$ outgoing entities.
The atomic facts $(e_{h,i},r_i^{\mathrm{in}},e_c)$ and
$(e_c,r_j^{\mathrm{out}},e_{t,j})$ generate the composite fact
$(e_{h,i},r_{ij}^{\mathrm{comp}},e_{t,j})$ through the shared rule
$r_i^{\mathrm{in}}\circ r_j^{\mathrm{out}}=r_{ij}^{\mathrm{comp}}$
(Figure~\ref{fig:exp-rule-composition}).
Different centers use disjoint entities but share the same rules, so
each rule applies across multiple factual instances. Each $(e_h,r_i)$
query has a unique true tail. We use
$d_{\mathrm{in}}=d_{\mathrm{out}}=d$, giving $d^2$ composition rules
and $2d$ atomic facts per complete center.

\begin{figure}[ht]
    \centering
    \includegraphics[width=0.75\linewidth]{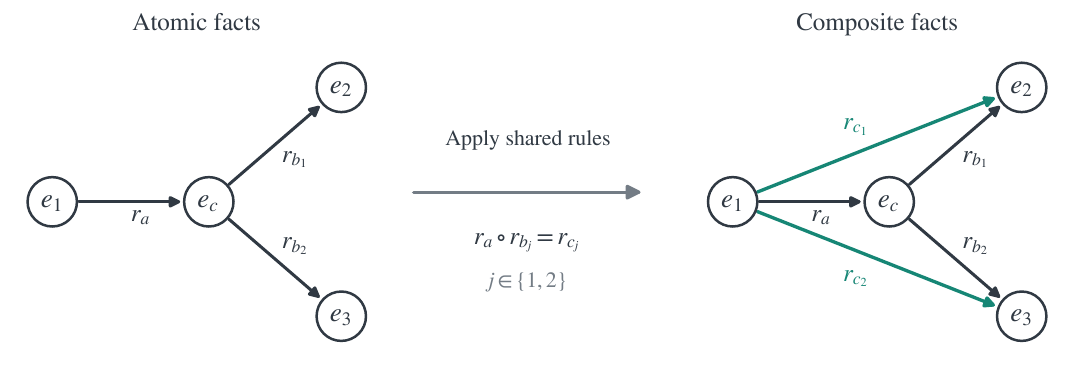}
    \caption{Shared composition rules map atomic facts to composite facts.}
    \label{fig:exp-rule-composition}
\end{figure}

All atomic facts enter training. We distinguish three center types
by their composite facts (Figure~\ref{fig:exp-training-test-construction}):
\begin{itemize}
[leftmargin=1.25em]
    \item \textbf{Full-training centers}: all composite facts enter training.
    \item \textbf{Atomic-only centers}: no composite facts enter training or testing.
    \item \textbf{Split-test centers}: half the composite facts enter training
    and the other half enter testing, with balanced rule coverage;
    each test fact is unseen during training, but both atomic premises are observed.
\end{itemize}

\begin{figure}[ht]
    \centering
    \includegraphics[width=0.8\linewidth]{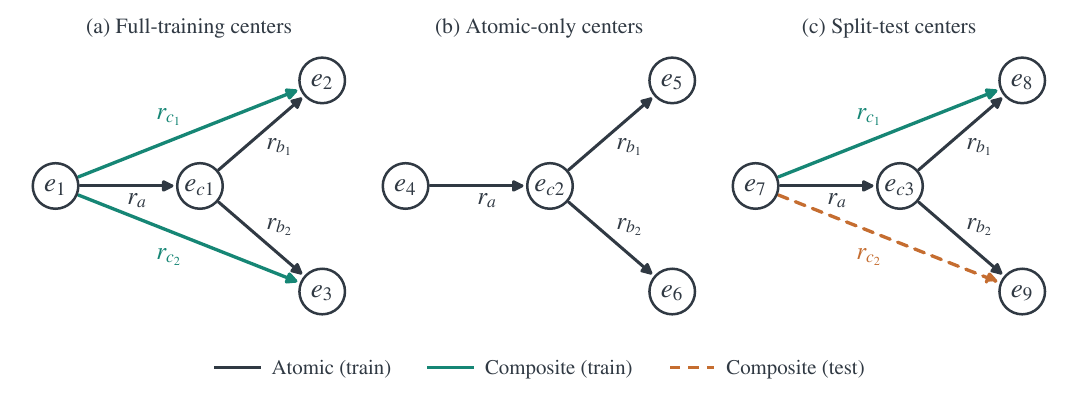}
    \caption{Training and test splits across the three center types. Solid edges enter training; dashed edges enter testing.}
    \label{fig:exp-training-test-construction}
\end{figure}

\subsection{Models and Training}

We use randomly initialized GPT-style decoder-only
Transformers~\citep{radford2019language} with the configurations in
Table~\ref{tab:exp-models}. Entities and relations are represented as
tokens, and each fact $z=(e_h,r_i,e_t)$ is encoded as
$(\mathtt{BOS},e_h,r_i,e_t)$ for autoregressive training, with loss
\begin{align*}
\mathcal L(w,z)
&=-\frac13\log p_w(e_h,r_i,e_t\mid\mathtt{BOS})\\
&=-\frac13\Bigl[\log p_w(e_h\mid\mathtt{BOS})
  +\log p_w(r_i\mid\mathtt{BOS},e_h)\\
&\hspace{4em}+\log p_w(e_t\mid\mathtt{BOS},e_h,r_i)\Bigr].
\end{align*}
Unless otherwise specified, the vocabulary contains 40,961 tokens:
32,768 entity tokens, 8,192 relation tokens, and $\mathtt{BOS}$.
The 1,024-rule graph uses 1,280 relation tokens instead, giving
a vocabulary of 34,049 tokens.
We train using AdamW with learning rate $10^{-3}$,
batch size 10,000, and weight decay $0.1$.
Composite test accuracy is the fraction of held-out composite facts
for which the highest-probability entity is the true tail.

\begin{table}[ht]
\centering
\caption{Transformer configurations with a 40,961-token vocabulary.}
\label{tab:exp-models}
\small
\begin{tabular}{ccccc}
\toprule
Parameters (M) & Layers & Hidden size & Attention heads & FFN size \\
\midrule
0.66 & 1 & 8 & 1 & 32 \\
1.31 & 1 & 16 & 1 & 64 \\
1.65 & 2 & 20 & 4 & 80 \\
1.98 & 2 & 24 & 4 & 96 \\
2.31 & 2 & 28 & 4 & 112 \\
2.65 & 2 & 32 & 4 & 128 \\
4.02 & 3 & 48 & 4 & 192 \\
8.31 & 4 & 96 & 4 & 384 \\
18.40 & 6 & 192 & 8 & 768 \\
\bottomrule
\end{tabular}
\end{table}

\subsection{Experimental Setup}

Following Section~\ref{sec:experiments-relational-design}, we examine
how composite-fact exposure and rule-learning burden interact with
model capacity to shape atomic hallucination. The two experiments
vary the training mixture and rule count, respectively, across model sizes.

\paragraph{Training mixture.}
With all $n_1$ atomic facts in training, we vary the number $n$ of
composite training facts to adjust $a=n/(n_1+n)$, keeping the atomic
facts, rules, and composite test set fixed. The total number of
atomic-only and full-training centers remains constant: adding an
atomic-only center's composite edges to training converts it into a
full-training center and increases $a$, while split-test centers remain
unchanged. We consider graphs with 4,096 and 1,024 rules,
with $a$ ranging over $[0.25,0.95]$ and $[0.23,0.94]$, respectively.

\paragraph{Rule count.}
At fixed $n_1$ and $a$, we vary the number of composition rules across
16, 256, and 6,400 while matching the atomic and composite training
counts and the composite test count. This changes the structural
learning burden without adding training facts. We compare these
settings at training proportions $a\in[0.40,0.55]$.
Table~\ref{tab:exp-data-overview} summarizes the datasets.

\begin{table}[ht]
\centering
\caption{Datasets for the training-mixture and rule-count experiments.}
\label{tab:exp-data-overview}
\small
\begin{tabular}{cccc}
\toprule
Experiment & Rules & Atomic train $n_1$ & Composite test \\
\midrule
Training mixture & 4,096 & 24,704 & 8,192 \\
 & 1,024 & 3,520 & 1,024 \\
Rule count & 16 / 256 / 6,400 & 24,704 & 6,400 \\
\bottomrule
\end{tabular}
\end{table}

\subsection{Additional Results}

Figure~\ref{fig:atomic-fpr-g3}(a) in the main text uses 4,096 rules,
$n_1=24{,}704$, and $a\in[0.25,0.95]$.
Figure~\ref{fig:g3-fpr} reports AtomicFP alongside composite test
accuracy for 1,024 rules, $n_1=3{,}520$, and $a\in[0.23,0.94]$.
Smaller models favor intermediate
mixtures for low AtomicFP, while larger models combine low FP with
near-perfect composite accuracy over a broader range. The mixture
thus affects both atomic factual discrimination and inference on
unseen composite facts.

\begin{figure}[t]
    \centering
    \includegraphics[width=\linewidth]{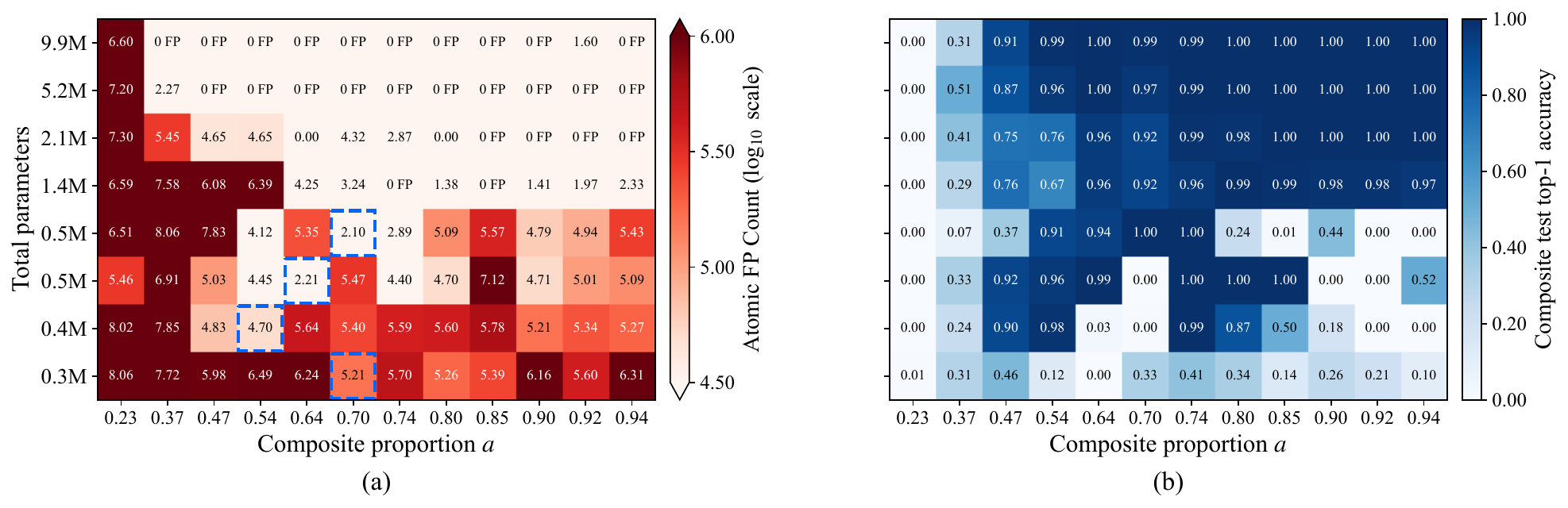}
    \caption{Training-mixture results with 1,024 rules: (a) AtomicFP at perfect fact recall; (b) composite test accuracy.}
    \label{fig:g3-fpr}
\end{figure}

Figure~\ref{fig:atomic-fpr-g3}(b) in the main text compares 16, 256,
and 6,400 rules at $n_1=24{,}704$ and $a=0.5$.
Figure~\ref{fig:4-rule-count} extends this comparison to other training
proportions $a\in[0.40,0.55]$, keeping $n_1$ fixed.
Across the displayed mixtures, models learning 6,400 rules
generally produce more atomic false positives than equally sized
models learning 16 or 256 rules. Larger models ease this trade-off:
the capacity needed to reach zero FP depends on both rule count and
training mixture.

\begin{figure}[t]
    \centering
    \includegraphics[width=\linewidth]{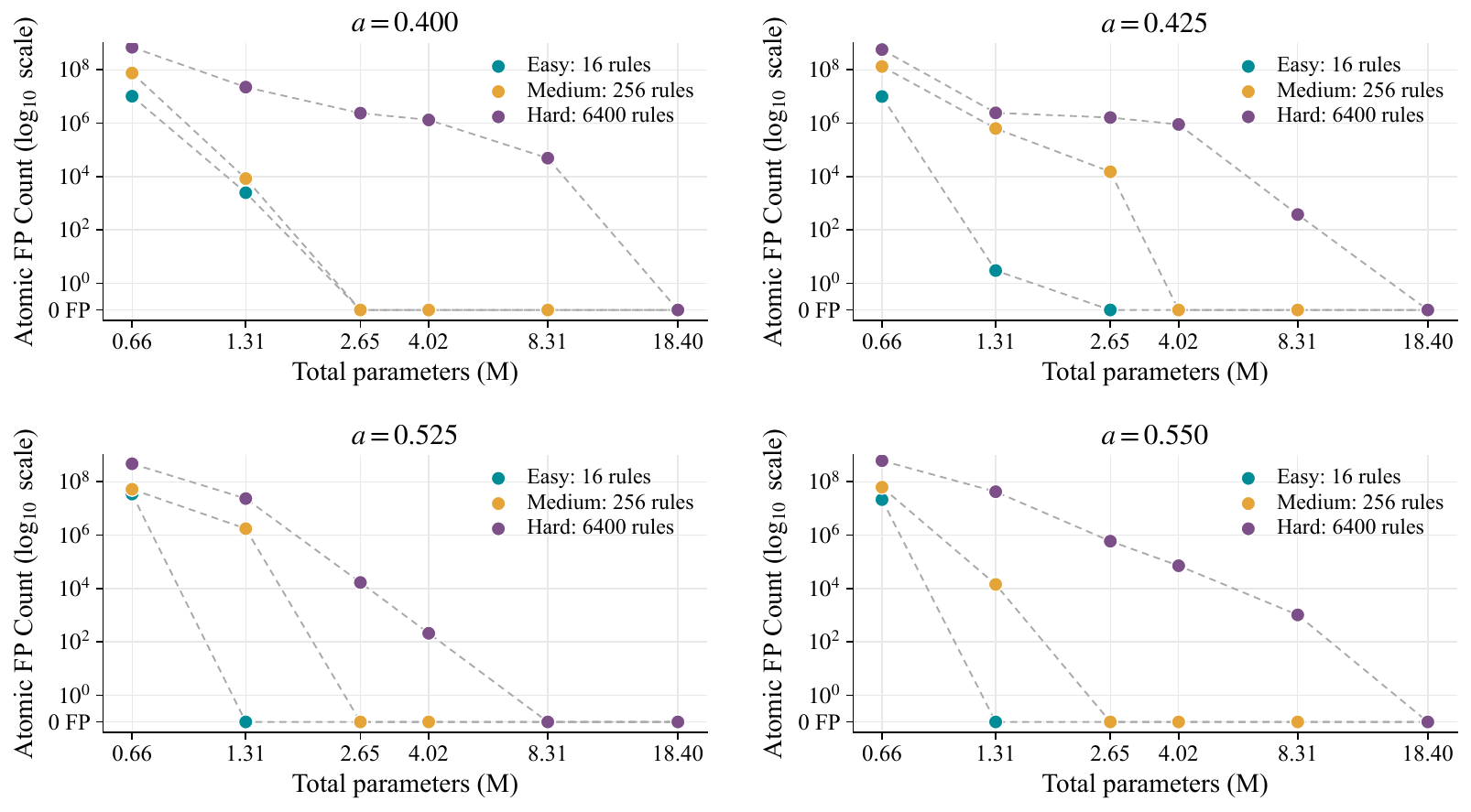}
    \caption{AtomicFP across model sizes and rule counts at four additional training mixtures.}
    \label{fig:4-rule-count}
\end{figure}

\section{Omitted Proofs}
\label{proofs}

\selectorJSidentity*

\begin{proof}
Since $U_i\sim\Bern(\rho)$, the marginal law of $T_i$ is
$M_\rho=\rho P+(1-\rho)Q$. By the definition of mutual information,
\begin{align*}
I(U_i;T_i)
&=\mathbb E_{U_i}\mathbb E_{T_i\mid U_i}
 \left[\log\frac{dP_{T_i\mid U_i}}{dP_{T_i}}(T_i)\right]\\
&=\rho\,\mathbb E_{T_i\mid U_i=1}
 \left[\log\frac{dP}{dM_\rho}(T_i)\right]
 +(1-\rho)\,\mathbb E_{T_i\mid U_i=0}
 \left[\log\frac{dQ}{dM_\rho}(T_i)\right]\\
&=\rho\kl{P}{M_\rho}+(1-\rho)\kl{Q}{M_\rho}\\
&=J_\rho(P,Q).
\end{align*}
\end{proof}

For $a,b\in[0,1]$,
let $d_{\mathrm{JS},\rho}(a,b)\defeq J_\rho\bigl(\Bern(a),\Bern(b)\bigr)$ be the binary JS divergence, which is also used in \cite{dong2025exactly}. We then have the following theorem.

\begin{thm}[Full statement of Theorem~\ref{thm:random-coordinate}]
\label{thm:random-coordinate-full}
Let $J\sim\Unif([N])$ be independent of
$(\widetilde Z,K,W)$, and define
$U\defeq\ones\{J\in K\}$ and $T\defeq T_J$. Then
\begin{align}
d_{\mathrm{JS},\rho}\!\left(\ex{}{\hat L_n},\ex{}{\hat L_m}\right)
&\le I(U;T)
\le I(U;T\mid\widetilde Z,J)
=\frac{1}{N}\sum_{i=1}^N I(U_i;T_i\mid\widetilde Z).
\label{eq:sharp-js}
\end{align}
In particular,
\begin{equation}
\abs{\mathcal E(\calA)}
\le
\sqrt{\frac{I(U;T\mid\widetilde Z,J)}{2\rho(1-\rho)}}.
\label{eq:random-coordinate-bound-full}
\end{equation}
\end{thm}

\begin{proof}
A direct calculation gives
\begin{align*}
\mathbb E[\hat L_n]
&=\frac1n\sum_{i=1}^N\mathbb E[U_iT_i]
 =\frac1N\sum_{i=1}^N\mathbb E[T_i\mid U_i=1]
 =\mathbb E[T\mid U=1],\\
\mathbb E[\hat L_m]
&=\frac1m\sum_{i=1}^N\mathbb E[(1-U_i)T_i]
 =\frac1N\sum_{i=1}^N\mathbb E[T_i\mid U_i=0]
 =\mathbb E[T\mid U=0].
\end{align*}
To obtain the binary projection, given $T\in[0,1]$, output one with
probability $T$ and zero otherwise, using randomness independent of $U$.
The conditional output laws are
$\Bern(\mathbb E[T\mid U=1])$ and
$\Bern(\mathbb E[T\mid U=0])$.
By Lemma~\ref{cmi-js}, the mutual information between $U$ and this
binary output is the weighted JS divergence between these laws.
Since this randomization uses only $T$, data processing gives
\[
d_{\mathrm{JS},\rho}\bigl(\mathbb E[\hat L_n],\mathbb E[\hat L_m]\bigr)
\leq I(U;T).
\]

Although $U$ is a function of $(J,K)$, uniform subset selection and
$K\indp\widetilde Z$ imply
$\mathbb P(U=1\mid\widetilde Z,J)=n/N=\rho$.
Thus $U\indp(\widetilde Z,J)$, and
\[
I(U;T)\leq I(U;T,\widetilde Z,J)
=I(U;T\mid\widetilde Z,J).
\]
Conditioning on the independent uniform index $J$ gives
\[
I(U;T\mid\widetilde Z,J)
=\frac1N\sum_{i=1}^N I(U_i;T_i\mid\widetilde Z),
\]
completing~\eqnref{eq:sharp-js}.

For the final bound, we establish the weighted JS--TV inequality.
For any $P,Q$ and $M_\rho=\rho P+(1-\rho)Q$, Pinsker's inequality gives
\begin{align*}
\TV(P,M_\rho)&=(1-\rho)\TV(P,Q)
 \leq\sqrt{\frac{\kl{P}{M_\rho}}2},\\
\TV(Q,M_\rho)&=\rho\TV(P,Q)
 \leq\sqrt{\frac{\kl{Q}{M_\rho}}2},
\end{align*}
where $\TV$ denotes total variation distance. Consequently,
\begin{align*}
J_\rho(P,Q)
&=\rho\kl{P}{M_\rho}+(1-\rho)\kl{Q}{M_\rho}\\
&\geq2\rho(1-\rho)^2\TV^2(P,Q)
     +2(1-\rho)\rho^2\TV^2(P,Q)\\
&=2\rho(1-\rho)\TV^2(P,Q).
\end{align*}
Applying this inequality to the two Bernoulli risk laws yields
\[
d_{\mathrm{JS},\rho}\bigl(\mathbb E[\hat L_n],\mathbb E[\hat L_m]\bigr)
\geq2\rho(1-\rho)
\bigl(\mathbb E[\hat L_n]-\mathbb E[\hat L_m]\bigr)^2.
\]
Combining this with~\eqnref{eq:sharp-js} and
$\mathcal E(\mathcal A)=\mathbb E[\hat L_m]-\mathbb E[\hat L_n]$
proves~\eqnref{eq:random-coordinate-bound-full}
and~\eqnref{eq:random-coordinate-bound}.
\end{proof}

\globalECMIBound*

\begin{proof}
For $t^N\in[0,1]^N$ and $k\subseteq[N]$ with $|k|=n$, let
\[
G(t^N,k)=\frac1m\sum_{i\notin k}t_i-\frac1n\sum_{i\in k}t_i.
\]
Fix $\widetilde Z=\tilde z$ and compare the true conditional joint
law of $(T^N,K)$ with the decoupled law
\[
Q_{\tilde z}=P_{T^N\mid\widetilde Z=\tilde z}\otimes P_K.
\]
Under $Q_{\tilde z}$, $K$ is a uniform subset independent of $T^N$,
and $\mathbb E_{Q_{\tilde z}}[G\mid T^N]=0$.
The sampling-without-replacement comparison and Hoeffding's
lemma~\citep{hoeffding1963probability} give, uniformly in $t^N$,
\begin{align*}
\log\mathbb E_K[e^{\lambda G(t^N,K)}]
&\leq\frac{\lambda^2N^2}{8nm^2},\\
\log\mathbb E_K[e^{\lambda G(t^N,K)}]
&\leq\frac{\lambda^2N^2}{8mn^2}.
\end{align*}
Indeed,
\[
G(t^N,K)
=\frac Nm\left(\frac1N\sum_i t_i-\frac1n\sum_{i\in K}t_i\right)
=\frac Nn\left(\frac1m\sum_{i\notin K}t_i-\frac1N\sum_i t_i\right),
\]
so the two bounds follow by applying the sample-mean exponential-moment
bound to the selected and unselected subsets, respectively.

Since $K\indp\widetilde Z$, the KL divergence between the true
conditional law and $Q_{\tilde z}$ equals
$I(T^N;K\mid\widetilde Z=\tilde z)$.
The Donsker--Varadhan variational inequality~\citep{donsker1983asymptotic},
applied to $\lambda G$ and $-\lambda G$, gives, for $\lambda>0$,
\[
\left|\mathbb E[G\mid\widetilde Z=\tilde z]\right|
\leq
\frac{I(T^N;K\mid\widetilde Z=\tilde z)}{\lambda}
+\frac{\lambda N^2}{8nm^2}.
\]
Optimizing over $\lambda>0$ yields
\[
\left|\mathbb E[G\mid\widetilde Z=\tilde z]\right|
\leq\frac Nm
\sqrt{\frac{I(T^N;K\mid\widetilde Z=\tilde z)}{2n}}.
\]
The second exponential-moment bound yields the complementary inequality
with $n,m$ interchanged. Averaging each bound over $\widetilde Z$
and applying Jensen's inequality gives
\[
|\mathcal E(\mathcal A)|
\leq\min\left\{
\frac Nm\sqrt{\frac{\eCMI^\ell_\rho(\mathcal A)}{2n}},
\frac Nn\sqrt{\frac{\eCMI^\ell_\rho(\mathcal A)}{2m}}
\right\},
\]
as claimed.
\end{proof}

\begin{rem}[Relation to existing CMI bounds]
\label{rem:existing-cmi-bounds}
Our fixed-cardinality selector is close to the leave-$m$-out
supersample setting~\citep{lu2026unified}.
When $N=2n$ and $m=n$, Theorem~\ref{thm:global-ecmi} reduces to
$\abs{\mathcal E(\calA)}
    \le
    \sqrt{
        \frac{2\eCMI^\ell_{1/2}(\calA)}{n}
    }$, which has the same numerical form as the evaluated pairwise-CMI bound~\citep{harutyunyan2021informationtheoretic,hellstrom2022a,wang2023tighter}. The underlying selector laws, however, are different. The original CMI construction selects exactly one instance from each of $n$ prescribed pairs using independent Bernoulli bits~\citep{steinke2020reasoning}, whereas our selector is uniform over all size-$n$ subsets of $[2n]$. It may therefore select both instances from
one pair and neither from another. Thus, $\rho=1/2$ recovers the balanced JS weighting, but not the same joint selector law or, in general, the same CMI quantity.
Our purpose here is to establish a common selector-information perspective, rather than to improve upon existing generalization bounds. For a comprehensive treatment, including tighter bounds, fast-rate results, and extensions beyond bounded losses, we refer to~\cite{hellstrom2025generalization}.
\end{rem}

\sparseJSKLLimit*

\begin{proof}
Hold the conditional score laws $P,Q$ fixed as $\rho\downarrow0$, and
let $M_\rho=\rho P+(1-\rho)Q$. Since
\[
\TV(M_\rho,Q)=\rho\TV(P,Q)\longrightarrow0,
\]
lower semicontinuity of KL divergence gives
\[
\kl{P}{Q}
\leq\liminf_{\rho\downarrow0}\kl{P}{M_\rho}.
\]
By nonnegativity of KL divergence,
\[
\frac{J_\rho(P,Q)}{\rho}
=\kl{P}{M_\rho}
+\frac{1-\rho}{\rho}\kl{Q}{M_\rho}
\geq\kl{P}{M_\rho}.
\]
Consequently,
\begin{equation}
\label{eq:sparse-js-kl-liminf}
\liminf_{\rho\downarrow0}\frac{J_\rho(P,Q)}{\rho}
\geq\kl{P}{Q}.
\end{equation}
If $\kl{P}{Q}=\infty$, this already proves the claim in the
extended-real sense.

Now suppose that $\kl{P}{Q}<\infty$, so $P\ll Q$.
Expanding the divergences and using $M_\rho=\rho P+(1-\rho)Q$ yields
\begin{align*}
\rho\kl{P}{Q}-J_\rho(P,Q)
&=\rho\int\log\frac{dM_\rho}{dQ}\,dP\\
&\quad +(1-\rho)\int\log\frac{dM_\rho}{dQ}\,dQ\\
&=\int\log\frac{dM_\rho}{dQ}\,dM_\rho\\
&=\kl{M_\rho}{Q}.
\end{align*}
Thus
\[
\frac{J_\rho(P,Q)}{\rho}
=\kl{P}{Q}-\frac{\kl{M_\rho}{Q}}{\rho}
\leq\kl{P}{Q}.
\]
Together with~\eqnref{eq:sparse-js-kl-liminf}, this proves
\[
\lim_{\rho\to0}\frac{J_\rho(P,Q)}{\rho}
=\kl{P}{Q}.
\]
Lemma~\ref{cmi-js} gives $I(U;T)=J_\rho(P,Q)$, completing the stated
limit. In the finite-KL case, the same identity also implies
\[
\frac{\kl{M_\rho}{Q}}{\rho}
=\kl{P}{Q}-\frac{J_\rho(P,Q)}{\rho}
\longrightarrow0,
\]
so $\kl{M_\rho}{Q}=o(\rho)$ follows from the limit rather than being
assumed.
\end{proof}

\FPRrandomfact*

\begin{proof}
The claim universe $\widetilde Z$ is fixed, so conditioning on it is
implicit. Recall that $U_i=\ones\{i\in K\}$, $U=U_J$, and $T=T_J$.
Since $K$ is uniform over the size-$n$ subsets of $[N]$,
$\mathbb P(U_i=1)=\binom{N-1}{n-1}/\binom{N}{n}=\rho$ for every $i$.
Thus $U\indp J$, and the independence of $J$ from $(K,T^N)$ gives
\begin{align*}
J_\rho(P,Q)
&=I(U;T)
\leq I(U;T,J)
=I(U;T\mid J)\\
&=\frac1N\sum_{i=1}^N I(U_i;T_i).
\end{align*}
To bound this sum by the information in the full score vector, note
that $(U_1,\ldots,U_N)$ determines $K$ and conversely. Conditional
entropy subadditivity and the fact that $T_i$ is a coordinate of $T^N$
therefore imply
\[
H(K\mid T^N)
=H(U_1,\ldots,U_N\mid T^N)
\leq\sum_{i=1}^N H(U_i\mid T^N)
\leq\sum_{i=1}^N H(U_i\mid T_i).
\]
Using $H(U_i)=h(\rho)$ and $H(K)=\log\binom{N}{n}$, we obtain
\begin{align*}
\sum_{i=1}^N I(U_i;T_i)
&=Nh(\rho)-\sum_{i=1}^N H(U_i\mid T_i)\\
&\leq Nh(\rho)-H(K\mid T^N)\\
&=I(T^N;K)+Nh(\rho)-\log\binom{N}{n}.
\end{align*}
Since $N\rho=n$, the assumed information budget yields
\begin{equation}
\label{eq:random-fact-score-information-bound}
\frac{J_\rho(P,Q)}{\rho}
\leq\frac{I(T^N;K)}{n}
+\frac{Nh(\rho)-\log\binom{N}{n}}{n}
\leq R+\frac{Nh(\rho)-\log\binom{N}{n}}{n}.
\end{equation}
The correction accounts for the dependence among the fixed-cardinality
membership indicators. As $n\to\infty$ and $\rho\to0$, we also have
$N-n\to\infty$, and Stirling's formula gives
\[
\log\binom{N}{n}
=Nh(\rho)-\frac12\log\bigl(2\pi n(1-\rho)\bigr)
+O\!\left(\frac1n+\frac1{N-n}\right).
\]
Consequently,
\[
\frac{Nh(\rho)-\log\binom{N}{n}}{n}
=\frac{\log\bigl(2\pi n(1-\rho)\bigr)}{2n}
+O\!\left(\frac1{n^2}+\frac1{n(N-n)}\right)
=o(1),
\]
so \eqnref{eq:random-fact-score-information-bound} gives
$J_\rho(P,Q)/\rho\leq R+o(1)$.

Thresholding is a post-processing of the score, so the data-processing inequality gives
\[
\frac{J_\rho(P,Q)}{\rho}
\geq \frac{J_\rho\bigl(\Bern(\alpha_\tau),\Bern(\beta_\tau)\bigr)}{\rho}.
\]
Combining this with the preceding bound and dropping the nonnegative
second term in the definition of $J_\rho$ yields
\[
\kl{\Bern(\alpha_\tau)}
{\Bern\bigl(\rho\alpha_\tau+(1-\rho)\beta_\tau\bigr)}
\leq
\frac{J_\rho\bigl(\Bern(\alpha_\tau),\Bern(\beta_\tau)\bigr)}{\rho}
\leq R+o(1).
\]
For any limiting recall--false-positive pair $(\alpha_\tau,\beta_\tau)$,
lower semicontinuity of KL divergence therefore gives
$\kl{\Bern(\alpha_\tau)}{\Bern(\beta_\tau)}\leq R$.
This argument also allows the score laws and threshold to vary along
the sparse sequence.
For $\alpha_\tau>0$, the exponential lower bound follows from
\[
\kl{\Bern(\alpha_\tau)}{\Bern(\beta_\tau)}
=-h(\alpha_\tau)-\alpha_\tau\log\beta_\tau
-(1-\alpha_\tau)\log(1-\beta_\tau)
\geq -h(\alpha_\tau)-\alpha_\tau\log\beta_\tau.
\]
Rearranging gives
$\beta_\tau\geq\exp\bigl(-(R+h(\alpha_\tau))/\alpha_\tau\bigr)$.
For $\alpha_\tau=1$, $h(1)=0$ gives $\beta_\tau\geq e^{-R}$.
\end{proof}

\begin{figure}[t]
    \centering
    \includegraphics[width=0.6\linewidth]{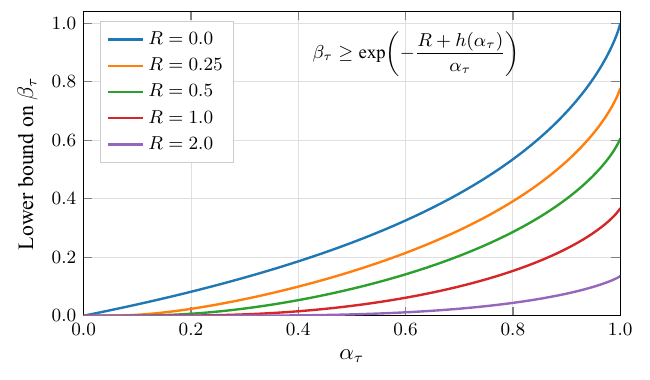}
    \caption{Lower bounds on the false-positive hallucination rate versus fact recall in the sparse random-fact setting, under a per-fact information budget $R$.}
    \label{fpr-recall-bound}
\end{figure}

\threeMIdecomposition*

\begin{proof}
Since $S$ is a deterministic function of $(B,\Theta,K)$ and $(B,\Theta,K)\to S\to W$ forms a Markov chain, the Markov property and the chain rule give
\begin{align*}
    I(W;S)=I(W;B,\Theta,K)=I(W;\Theta)
      +I(W;B\mid\Theta)
      +I(W;K\mid B,\Theta).
\end{align*}
Moreover, $K$ is independent of $(B,\Theta)$, and the
conditional distribution of $W$ depends on $(B,\Theta,K)$
only through $\widetilde Z_K$.
Consequently,
\[
    (W,K)\indp (B,\Theta)\mid\widetilde Z.
\]
The first equality below holds because $\widetilde Z$ is a deterministic function of
$(B,\Theta)$; the second follows from the conditional independence above:
\[
    I(W;K\mid B,\Theta)
    =
    I(W;K\mid B,\Theta,\widetilde Z)
    =
    I(W;K\mid\widetilde Z).
\]
Substituting this identity proves the decomposition.
\end{proof}

\begin{NoHyper}
\jointInformationBudget*
\end{NoHyper}
\begin{proof}
We first lower-bound the information used for rule recovery and sample selection, then relate the remaining atomic-fact budget to the false-positive rate under perfect recall.
By Fano's inequality,
\begin{align}
I(W;\Theta)
&=
H(\Theta)-H(\Theta\mid W) \notag\\
&\geq
\log N_\Theta
-h(\varepsilon_\Theta)
-\varepsilon_\Theta\log(N_\Theta-1)
=
C_{\mathrm{rule}}.
\label{eq:proof-rule-lower-bound}
\end{align}

Let
\[
\eCMI^\ell_\rho(\mathcal A)
=
I(T^N;K\mid\widetilde Z)
\]
denote the evaluated selector-density CMI associated with the
bounded loss. The global generalization bound gives
\[
|\mathcal E(\mathcal A)|
\leq
\min\left\{
\frac{N}{m}
\sqrt{\frac{\eCMI^\ell_\rho(\mathcal A)}{2n}},
\frac{N}{n}
\sqrt{\frac{\eCMI^\ell_\rho(\mathcal A)}{2m}}
\right\}.
\]
Inverting both inequalities yields
\begin{align}
\eCMI^\ell_\rho(\mathcal A)
&\geq
\max\left\{
\frac{2nm^2}{N^2},
\frac{2mn^2}{N^2}
\right\}
|\mathcal E(\mathcal A)|^2 \notag\\
&=
\frac{2nm\max\{n,m\}}{N^2}
|\mathcal E(\mathcal A)|^2
=
C_{\mathrm{gen}}.
\label{eq:proof-gen-lower-bound}
\end{align}
Since $T^N$ is obtained from $(W,\widetilde Z)$, the data processing
inequality gives
\[
I(W;K\mid\widetilde Z)
\geq
I(T^N;K\mid\widetilde Z)
\geq
C_{\mathrm{gen}}.
\]

By Lemma~\ref{3MI},
\[
I(W;S)
=
I(W;\Theta)
+
I(W;B\mid\Theta)
+
I(W;K\mid\widetilde Z).
\]
Combining this decomposition with $I(W;S)\leq C$,
\eqnref{eq:proof-rule-lower-bound}, and
\eqnref{eq:proof-gen-lower-bound}, we obtain
\begin{equation}
I(W;B\mid\Theta)
\leq
C-C_{\mathrm{rule}}-C_{\mathrm{gen}}.
\label{eq:proof-residual-fact-budget}
\end{equation}

It remains to lower-bound the information required for perfect
atomic recall. Let $K_1$ be the atomic-fact selector, so that
\[
K_1
\sim
\operatorname{Unif}
\left\{
A\subseteq[N_1]:|A|=n_1
\right\},
\qquad
\rho_1=\frac{n_1}{N_1}.
\]
Because the atomic candidate pool is fixed and contains distinct
claims, $B$ and $K_1$ are in one-to-one correspondence. Hence
\begin{equation}
I(W;B\mid\Theta)
=
I(W;K_1\mid\Theta).
\label{eq:proof-B-K1-equivalence}
\end{equation}

For a threshold $\tau$ and the model's factuality score $T_j$ of atomic claim $b_j$, define the set of accepted atomic claims by
\[
\widehat K_1
\coloneqq
\{j\in[N_1]:T_j\geq\tau\},
\qquad
L\coloneqq|\widehat K_1|.
\]
The fact recall satisfies
\[
\alpha_\tau
=
\frac{1}{n_1}
\mathbb E\!\left[
|K_1\cap\widehat K_1|
\right].
\]
Therefore, under perfect atomic recall, $\alpha_\tau=1$,
\[
\mathbb E\!\left[
|K_1\setminus\widehat K_1|
\right]
=
n_1(1-\alpha_\tau)
=
0.
\]
Since $|K_1\setminus\widehat K_1|$ is nonnegative, it follows that
\begin{equation}
K_1\subseteq\widehat K_1
\qquad\text{almost surely}.
\label{eq:proof-perfect-recall-containment}
\end{equation}
This is the step at which perfect recall is essential.

The false-positive hallucination rate is
\[
\beta_\tau
=
\frac{1}{N_1-n_1}
\mathbb E\!\left[
|\widehat K_1\setminus K_1|
\right].
\]
By \eqnref{eq:proof-perfect-recall-containment},
\[
L
=
n_1+|\widehat K_1\setminus K_1|
\qquad\text{almost surely},
\]
and consequently
\begin{equation}
\mathbb E[L]
=
n_1+(N_1-n_1)\beta_\tau.
\label{eq:proof-accepted-size}
\end{equation}

We now bound $H(K_1\mid W,\Theta)$ by counting the $n_1$-subsets compatible with the accepted set.
Since $K_1$ is independent of $\Theta$ and is uniform over all
$n_1$-subsets of $[N_1]$,
\[
H(K_1\mid\Theta)
=
\log\binom{N_1}{n_1}.
\]
Moreover, conditional on $(W,\Theta)$, the accepted set
$\widehat K_1$ is known. By
\eqnref{eq:proof-perfect-recall-containment}, $K_1$ must be an
$n_1$-subset of $\widehat K_1$. Thus,
\[
H(K_1\mid W,\Theta)
\leq
\mathbb E\!\left[
\log\binom{L}{n_1}
\right].
\]
Using \eqnref{eq:proof-B-K1-equivalence}, we obtain
\begin{align}
I(W;B\mid\Theta)
&\geq
\log\binom{N_1}{n_1}
-
\mathbb E\!\left[
\log\binom{L}{n_1}
\right] \notag\\
&=
\mathbb E\!\left[
\log
\frac{\binom{N_1}{n_1}}
     {\binom{L}{n_1}}
\right].
\label{eq:proof-combinatorial-information}
\end{align}

For every $n_1\leq L\leq N_1$,
\begin{align}
\frac{\binom{N_1}{n_1}}{\binom{L}{n_1}}
&=
\prod_{j=0}^{n_1-1}
\frac{N_1-j}{L-j}
\geq
\left(\frac{N_1}{L}\right)^{n_1}.
\end{align}
Therefore,
\begin{align}
I(W;B\mid\Theta)
&\geq
n_1\,
\mathbb E\!\left[
\log\frac{N_1}{L}
\right] \notag\\
&\geq
n_1
\log\frac{N_1}{\mathbb E[L]} \notag\\
&=
n_1
\log
\frac{1}
{\rho_1+(1-\rho_1)\beta_\tau},
\label{eq:proof-fact-information-lower-bound}
\end{align}
where the second inequality follows from Jensen's inequality and
the last equality follows from
\eqnref{eq:proof-accepted-size}.

Combining \eqnref{eq:proof-residual-fact-budget} and
\eqnref{eq:proof-fact-information-lower-bound} gives
\[
n_1
\log
\frac{1}
{\rho_1+(1-\rho_1)\beta_\tau}
\leq
C-C_{\mathrm{rule}}-C_{\mathrm{gen}}.
\]
Equivalently,
\[
\rho_1+(1-\rho_1)\beta_\tau
\geq
\exp\left(
-\frac{
C-C_{\mathrm{rule}}-C_{\mathrm{gen}}
}{n_1}
\right),
\]
and hence the following finite-density bound holds:
\begin{equation}
\beta_\tau
\geq
\max\left\{
0,\,
\frac{
\exp\left(
-\frac{
C-C_{\mathrm{rule}}-C_{\mathrm{gen}}
}{n_1}
\right)
-\rho_1
}
{1-\rho_1}
\right\}.
\label{eq:proof-finite-density-fpr}
\end{equation}

Finally, taking the sparse atomic-fact limit $\rho_1\to0$ in
\eqnref{eq:proof-finite-density-fpr} yields
\[
\beta_\tau
\geq
\exp\left(
-\frac{
C-C_{\mathrm{rule}}-C_{\mathrm{gen}}
}{n_1}
\right),
\]
which proves the result.
\end{proof}

\gaussianRecallLimit*

\begin{proof}[Proof of Theorem~\ref{thm:gaussian-recall}]
Throughout the proof, all probabilities and expectations are conditional
on $\Theta=\theta>0$, unless specified otherwise. In particular,
$\sigma_{\mathrm c}$ is fixed. We assume $\sigma_B>0$ and use the independence of $K$ from
$(\Theta,B,X)$.

\medskip
\noindent\textit{Convergence of the learned parameter.}
Let
\[
\overline Z_N\coloneqq\frac1N\sum_{i=1}^{n_1}\sum_{j=1}^{n_2}
\widetilde Z_{ij},
\qquad
\overline Z_{K,N}\coloneqq\frac1n\sum_{i=1}^{n_1}\sum_{j=1}^{n_2}
K_{ij}\widetilde Z_{ij}
\]
be the averages of the full and selected composite pools.
Since $\widetilde Z_{ij}=B_i+X_j$, the laws of large numbers give
\[
\overline Z_N
=\frac1{n_1}\sum_{i=1}^{n_1}B_i
+\frac1{n_2}\sum_{j=1}^{n_2}X_j
\xrightarrow{\mathbb P}\theta.
\]
Conditional on $(B,X)$, the selected entries form a simple random sample
without replacement from the $N$ composite values. Hence, for $N>1$,
\[
\mathbb E[\overline Z_{K,N}\mid B,X]=\overline Z_N,\qquad
\operatorname{Var}(\overline Z_{K,N}\mid B,X)
=\frac{1-n/N}{n(N-1)}
\sum_{i=1}^{n_1}\sum_{j=1}^{n_2}
(\widetilde Z_{ij}-\overline Z_N)^2.
\]
The empirical variance of the full Cartesian pool is the sum of the
empirical variances of the atomic values and the rule instances: the
cross term vanishes when summed over all $(i,j)$. Taking expectations,
\[
\mathbb E\!\left[\frac1N\sum_{i,j}
(\widetilde Z_{ij}-\overline Z_N)^2\right]
=\left(1-\frac1{n_1}\right)\sigma_B^2
+\left(1-\frac1{n_2}\right)\sigma_\theta^2
\le\sigma_{\mathrm c}^2.
\]
Consequently,
\[
\mathbb E[(\overline Z_{K,N}-\overline Z_N)^2]
\le\frac{N\sigma_{\mathrm c}^2}{n(N-1)}\longrightarrow0.
\]
Thus $\overline Z_{K,N}\xrightarrow{\mathbb P}\theta$. It follows that
\[
\widehat W_N
=(1-a_N)\frac1{n_1}\sum_{i=1}^{n_1}B_i
+a_N\overline Z_{K,N}
\xrightarrow{\mathbb P}a\theta.
\]
The assumption $W_N-\widehat W_N\xrightarrow{\mathbb P}0$ now yields
$W_N\xrightarrow{\mathbb P}a\theta$. No independence or distributional
assumption on $W_N-\widehat W_N$ is needed.

\medskip
\noindent\textit{Convergence of the two recalls.}
Define the empirical distribution functions of the atomic values, rule
instances, and full composite pool, respectively, by
\begin{align*}
F_N^{\mathrm{ato}}(x)&\coloneqq\frac1{n_1}\sum_{i=1}^{n_1}
\mathbf 1\{B_i\le x\},\qquad
F_{X,N}(x)\coloneqq\frac1{n_2}\sum_{j=1}^{n_2}\mathbf 1\{X_j\le x\},\\
F_N^{\mathrm{comp}}(x)&\coloneqq\frac1N\sum_{i=1}^{n_1}\sum_{j=1}^{n_2}
\mathbf 1\{\widetilde Z_{ij}\le x\}.
\end{align*}
The Glivenko--Cantelli theorem applied separately to the two marginal
samples gives
\[
\sup_x\left|F_N^{\mathrm{ato}}(x)-\Phi(x/\sigma_B)\right|
\xrightarrow{\mathbb P}0,\qquad
\sup_x\left|F_{X,N}(x)-\Phi((x-\theta)/\sigma_\theta)\right|
\xrightarrow{\mathbb P}0.
\]
These conclusions hold jointly by the union bound and do not require
independence between the two samples.

The Cartesian structure of the composite pool gives
\[
F_N^{\mathrm{comp}}(x)=\frac1{n_2}\sum_{j=1}^{n_2}F_N^{\mathrm{ato}}(x-X_j).
\]
Replacing $F_N^{\mathrm{ato}}$ by its Gaussian limit in this average
incurs an error bounded by
$\sup_y|F_N^{\mathrm{ato}}(y)-\Phi(y/\sigma_B)|$.
For the resulting average, the identity
\[
\frac1{n_2}\sum_{j=1}^{n_2}\Phi\!\left(\frac{x-X_j}{\sigma_B}\right)
=\int_{\mathbb R}F_{X,N}(x-b)\frac1{\sigma_B}\phi(b/\sigma_B)\,db
\]
and the Gaussian convolution formula yield
\begin{align*}
\sup_x\left|F_N^{\mathrm{comp}}(x)
-\Phi\!\left(\frac{x-\theta}{\sigma_{\mathrm c}}\right)\right|
&\le\sup_x\left|F_N^{\mathrm{ato}}(x)-\Phi(x/\sigma_B)\right|\\
&\quad+\sup_x\left|F_{X,N}(x)-\Phi((x-\theta)/\sigma_\theta)\right|
\xrightarrow{\mathbb P}0.
\end{align*}
Thus the empirical distribution of the composite pool converges to
$\mathcal N(\theta,\sigma_{\mathrm c}^2)$ even without assuming that
$B$ and $X$ are independent given $\Theta$.

To evaluate recall at the data-dependent center, let $F_N$ denote either
$F_N^{\mathrm{ato}}$ or $F_N^{\mathrm{comp}}$, and let $F$ denote its
continuous Gaussian limit established above. The corresponding empirical
proportion in $[W_N-r,W_N+r]$ equals
$F_N(W_N+r)-F_N((W_N-r)^-)$, where $F_N(t^-)$ is the left limit at $t$.
Continuity of $F$ gives the pathwise bound
\[
\left|F_N(W_N+r)-F_N((W_N-r)^-)
-\bigl(F(W_N+r)-F(W_N-r)\bigr)\right|
\le2\sup_x|F_N(x)-F(x)|.
\]
Combining uniform convergence with $W_N\xrightarrow{\mathbb P}a\theta$
therefore proves, for every fixed $r\ge0$,
\begin{align*}
\alpha_N^{\mathrm{ato}}(r)
&\xrightarrow{\mathbb P}
\Phi\!\left(\frac{a\theta+r}{\sigma_B}\right)
-\Phi\!\left(\frac{a\theta-r}{\sigma_B}\right)
=\alpha^{\mathrm{ato}}(r;a),\\
\alpha_N^{\mathrm{comp}}(r)
&\xrightarrow{\mathbb P}
\Phi\!\left(\frac{a\theta+r-\theta}{\sigma_{\mathrm c}}\right)
-\Phi\!\left(\frac{a\theta-r-\theta}{\sigma_{\mathrm c}}\right)
=\alpha^{\mathrm{comp}}(r;a).
\end{align*}
The final identities use $\Phi(-x)=1-\Phi(x)$ and also cover $r=0$.
No independence between $W_N$ and the facts is used. Since recall is
computed over the full composite pool, no condition on $N-n$ is needed.

\medskip
\noindent\textit{Convergence of the false-positive rate.}
For this paragraph, probabilities involving $Z$ refer only to the fresh
draw $Z\sim\mathcal N(W_N,1)$ with $(B,X,W_N)$ held fixed, as in the
definition of $\beta_N$.
Define $\varepsilon_N\coloneqq\mathbb P(Z\in\mathcal T_N)$.
The total length of $\mathcal T_N$ is at most
$2(n_1+N)\delta_N$, and the density of $Z$ is everywhere at most
$1/\sqrt{2\pi}$. Since $n_1\le N$, it follows that
\[
0\le\varepsilon_N
\le\frac{2(n_1+N)\delta_N}{\sqrt{2\pi}}
\le\frac{4N\delta_N}{\sqrt{2\pi}}
\longrightarrow0.
\]
Moreover, $\mathbb P(|Z-W_N|\le r)=2\Phi(r)-1=\beta(r)$ for every
realized $W_N$. For all sufficiently large $N$, the definition of
$\beta_N$ therefore gives
\begin{align*}
|\beta_N(r)-\beta(r)|
&=\frac{\left|\beta(r)\varepsilon_N
-\mathbb P\bigl(|Z-W_N|\le r,\ Z\in\mathcal T_N\bigr)\right|}
{1-\varepsilon_N}\\
&\le\frac{\varepsilon_N}{1-\varepsilon_N}
\longrightarrow0.
\end{align*}
The bound is uniform over the realized facts and learned parameter.
It proves the stated convergence in probability and completes the proof.
\end{proof}

\optimalMixtureRecall*

\begin{proof}[Proof of Corollary~\ref{cor:optimal-mixture-recall}]
For each $a\in[0,1]$ and $g\in\{\mathrm{ato},\mathrm{comp}\}$,
the function $r\mapsto\alpha^g(r;a)$ is continuous and strictly
increasing from $0$ to $1$. Consequently, $r_g(a)$ is finite and is
the unique radius satisfying $\alpha^g(r_g(a);a)=\alpha_0$.
In particular, $r_{\alpha_0}(a)$ satisfies both recall constraints.
The recall formulas are continuously differentiable and their
derivatives with respect to $r$ are positive, so the implicit function
theorem also makes $r_g(a)$ continuous in $a$.
Thus $\beta_{\alpha_0}(a)$ is continuous on $[0,1]$ and attains its
minimum there.

Fix any $a\in[0,1]$ and any radius $r$ satisfying both recall
constraints. Since $\Phi(x)\le1$, the limiting recall formulas imply
\[
\alpha_0\le\alpha^{\mathrm{ato}}(r;a)
\le\Phi\!\left(\frac{r-a\theta}{\sigma_B}\right),
\qquad
\alpha_0\le\alpha^{\mathrm{comp}}(r;a)
\le\Phi\!\left(\frac{r-(1-a)\theta}{\sigma_{\mathrm c}}\right).
\]
Applying the increasing function $\Phi^{-1}$ yields
\[
r\ge a\theta+\sigma_B\Phi^{-1}(\alpha_0),
\qquad
r\ge(1-a)\theta+\sigma_{\mathrm c}\Phi^{-1}(\alpha_0).
\]
Adding these inequalities and dividing by two eliminates $a$:
\[
r\ge\frac{\theta}{2}
+\frac{\sigma_B+\sigma_{\mathrm c}}{2}\Phi^{-1}(\alpha_0).
\]
Apply this bound to $r=r_{\alpha_0}(a)$.
Because $\beta(r)=2\Phi(r)-1$ is increasing, every $a\in[0,1]$ satisfies
\[
\beta_{\alpha_0}(a)
\ge
2\Phi\!\left(
\frac{\theta}{2}
+\frac{\sigma_B+\sigma_{\mathrm c}}{2}\Phi^{-1}(\alpha_0)
\right)-1.
\]
Taking the minimum over $a$ proves
\eqnref{eq:gaussian-optimal-fpr-lower}.
\end{proof}

\end{document}